\documentclass[11pt]{article}

\usepackage[letterpaper,margin=1in]{geometry}
\usepackage{amsmath, amssymb, amsthm, amsfonts}
\usepackage{bbm}
\usepackage{bm}
\usepackage{complexity}

\usepackage{subcaption}

\usepackage{ifthen}
\usepackage{comment}
\usepackage{tikz}
\usetikzlibrary{positioning,decorations.pathreplacing}
\usepackage{comment}

\usepackage{cite}
\usepackage{appendix}
\usepackage{graphicx}	
\usepackage{color}
\usepackage{algorithm}
\usepackage[noend]{algpseudocode}
\usepackage{epstopdf}
\usepackage{wrapfig}
\usepackage{paralist}
\usepackage[textsize=tiny]{todonotes}

\usepackage{framed}
\usepackage[framemethod=tikz]{mdframed}
\usepackage[bottom]{footmisc}
\usepackage{enumitem}
\setitemize{noitemsep,topsep=3pt,parsep=3pt,partopsep=3pt}
\usepackage[font=small]{caption}
\usepackage{xspace}

\definecolor{darkgreen}{rgb}{0,0.5,0}
\usepackage{hyperref}
\hypersetup{
    unicode=false,          
    colorlinks=true,        
    linkcolor=red,          
    citecolor=darkgreen,        
    filecolor=magenta,      
    urlcolor=cyan           
}
\usepackage[capitalize]{cleveref}

\newtheorem{theorem}{Theorem}[section]
\newtheorem{lemma}[theorem]{Lemma}
\newtheorem{meta-theorem}[theorem]{Meta-Theorem}

\newtheorem{corollary}[theorem]{Corollary}

\newtheorem{definition}[theorem]{Definition}

\newcommand{\eps}{\varepsilon}

\newcommand{\ED}{\textnormal{ED}}

\newcommand{\LCS}{\textnormal{LCS}}

\newcommand{\Decode}{\textnormal{Dec}}
\newcommand{\Encode}{\textnormal{Enc}}

\newcommand{\exclude}[1]{}

\newcommand{\FullOrShort}{full}

\ifthenelse{\equal{\FullOrShort}{full}}{

  \newcommand{\fullOnly}[1]{#1}	
	\newcommand{\tempfullOnly}[1]{#1}
  \newcommand{\shortOnly}[1]{}
  
  }{

    \newcommand{\fullOnly}[1]{}
		\newcommand{\tempfullOnly}[1]{}
		\newcommand{\shortOnly}[1]{#1}
    \newcommand{\IncludePictures}[1]{}
   
  }

\newcommand{\STOConly}[1]{}
\newcommand{\noSTOC}[1]{#1}

\begin{document}

\date{}

\title{Near-Linear Time Insertion-Deletion Codes and (1+$\varepsilon$)-Approximating Edit Distance via Indexing}	

\author{Bernhard Haeupler\thanks{Supported in part by NSF grants CCF-1618280, CCF-1814603, CCF-1527110, NSF CAREER award CCF-1750808, and a Sloan Research Fellowship.}\\Carnegie Mellon University\\ \texttt{haeupler@cs.cmu.edu}  \and
Aviad Rubinstein\thanks{Supported in part by a Robert N. Noyce Family Faculty Fellowship.}\\Stanford University\\ \texttt{aviad@cs.stanford.edu} \and
Amirbehshad Shahrasbi\footnotemark[1]\\Carnegie Mellon University\\ \texttt{shahrasbi@cs.cmu.edu}}

\allowdisplaybreaks
	
\maketitle

\begin{abstract}
We introduce \emph{fast-decodable indexing schemes for edit distance} which can be used to speed up edit distance computations to near-linear time if one of the strings is indexed by an indexing string $I$. In particular, for every length $n$ and every $\eps >0$, one can in near linear time construct a string $I \in \Sigma'^n$ with $|\Sigma'| = O_{\eps}(1)$, such that, indexing any string $S \in \Sigma^n$, symbol-by-symbol, with $I$ results in a string $S' \in \Sigma''^n$ where $\Sigma'' = \Sigma \times \Sigma'$ for which edit distance computations are easy, i.e., one can compute a 
$(1+\eps)$-approximation of the edit distance between $S'$ and any other string in $O(n \poly(\log n))$ time.

Our indexing schemes can be used to improve the decoding complexity of state-of-the-art error correcting codes for insertions and deletions. In particular, they lead to near-linear time decoding algorithms for the insertion-deletion codes of [Haeupler, Shahrasbi; STOC `17] and faster decoding algorithms for list-decodable insertion-deletion codes of [Haeupler, Shahrasbi, Sudan; ICALP `18]. Interestingly, the latter codes are a crucial ingredient in the construction of fast-decodable indexing schemes.

\end{abstract}
	
\setcounter{page}{0}
\thispagestyle{empty}

\newpage
\section{Introduction}

Error correcting codes have revolutionized how information is stored and transmitted. The main two parameters of interest for an error correcting code are: (1) its \emph{rate-distance trade-off}, i.e., how much redundancy is added in comparison to how many errors the code can correct and (2) its \emph{computational efficiency}, i.e., how fast one can encode or decode. Ideal families of codes, so called near-MDS codes, over large finite alphabets achieve an (almost) perfect rate distance trade-off, i.e., attain a rate of $1-\delta-\eps$ while correcting any $\delta/2$ fraction of errors (or $\delta$ fraction of erasures, insertions, or deletions) for any $\delta \in (0,1)$ and are encodable and decodable in (near) linear time.

\STOConly{\vspace{-1mm}}
\subsection{(Near) Linear-Time Codes}

The seminal works of Shannon, Hamming, and others in the late 40s and early 50s established a good understanding of the optimal rate/distance tradeoffs achievable existentially and over the next decades, near-MDS codes achieving at least polynomial time decoding and encoding procedures were put forward. Since then, lowering the computational complexity has been an important goal of coding theory. Particularly, the 90s saw a big push, spearheaded by Spielman, to achieve (near) linear coding complexities: in a breakthrough in 1994, Sipser and Spielman~\cite{sipser1994expander} introduced expander codes and derived linear codes with some constant distance and rate that are decodable (but not encodable) in linear time. In 1996 Spielman~\cite{spielman1996linear} build upon these codes to derive asymptotically good error correcting codes that are encodable and decodable in linear time. As for codes with better rate distance trade-off, Alon et al.~\cite{alon1995linear} obtained near-MDS error correcting codes that were decodable from erasures in linear time. 
Finally, in 2004, Guruswami and Indyk~\cite{guruswami2005linear} provided near-MDS error correcting codes for any rate than can be decoded in linear time from any combination of symbol erasures and symbol substitutions.

\STOConly{\vspace{-1mm}}
\subsection{Codes for Insertions and Deletions}
Similar questions on communication and computational efficiencies hold for synchronization codes, i.e., codes that correct from symbol insertions and symbol deletions. As a matter of fact, an analogous flow of progress can be recognized for synchronization codes. The study of synchronization codes started with the work of Levenshtein~\cite{Levenshtein65} in the 60s. In 1999, Schulman and Zuckerman~\cite{schulman1999asymptotically} gave the first (efficient) synchronization code with constant distance and rate. Only recently, synchronization codes with stronger communication efficiency have been found. Guruswami et al.~\cite{guruswami2017deletion,guruswami2016efficiently} introduced the first synchronization codes in the asymptotically small or large noise regimes by giving efficient codes which achieve a constant rate for noise rates going to one and codes which provide a rate going to one for an asymptotically small noise rate. Last year, Haeupler and Shahrasbi~\cite{haeupler2017synchronization} were able to finally achieve efficient synchronization codes with the optimal (near-MDS) rate/distance tradeoff, for any rate and distance. These codes build on a novel tool called \emph{synchronization strings} which are also used in~\cite{haeupler2018synchronization4} to design efficient list-decodable codes from insertions and deletions.

All of the codes mentioned so far have decoders with large polynomial complexity between $\Omega(n^2)$ and $O(n^{O(1/\eps)})$. The only known insertion-deletion codes with subquadratic time decoders are given in \cite{haeupler2017synchronization3}. Unfortunately, these codes only work for $\delta\in\left(0,\frac{1}{3}\right)$ fraction of errors while achieving a rate of $1-3\delta-\eps$ (instead of the desired $1-\delta-\eps$).

In this work, we take the natural next step and address the problem of finding near-linear time encodable/decodable (near-MDS) codes for insertions and deletions.

\subsection{Quadratic Edit Distance Computation Barrier}

Many of the techniques developed for constructing efficient regular error correcting codes also apply to synchronization strings. Indeed, the synchronization string based constructions in \cite{haeupler2017synchronization,haeupler2017synchronization2,haeupler2017synchronization3,haeupler2018synchronization4} show that this can largely be done in a black-box manner. However, there is a serious new barrier that arises  in the setting of synchronization errors if one tries to push computational complexities below $n^2$. This barrier becomes apparent once one notices that decoding an error correcting code is essentially doing a distance minimization where the appropriate distance in the synchronization setting is the edit distance%
\footnote{We define the {\em edit distance} between two strings $S,S'$ as the minimum number of character insertions and deletions required to transform $S$ to $S'$. Note that this is slightly different (but closely related) to the more standard definition which also allows character substitutions.}.
As we discuss below, merely computing the edit distance between two strings (the input and a candidate output of the decoder) in subquadratic time is a well known hard problem.
An added dimension of challenge in our setting is that we must first select the candidate decoder outputs among exponentially many codewords.

A simple algorithm for computing the edit distance of two given strings is the classic Wagner-Fischer 
dynamic programming algorithm that runs in quadratic time. 
Improving the running time of this simple algorithm has been a central open problem in computer science for decades (e.g.~\cite{Knuth72}).
Yet to date, only a slightly faster algorithm ($O(n^2/\log^2 n)$) due to Masek and Paterson~\cite{masek1980faster} is known. 
Furthermore, a sequence of complexity breakthroughs from recent years suggests that a near-quadratic running time may in fact be optimal~\cite{AWW14-Local_Alignment, backurs2015edit, ABW15-LCS, BK15-LCS, AHWW16-polylog_shaved} (under the Strong Exponential Time Hypothesis (SETH) or related assumptions).
In order to obtain subquadratic running times, computer scientists have considered two directions: moving beyond worst-case instances, and allowing approximations. 
 
\subsubsection*{Beyond worst case}
Edit distance computation is known to be easier in several special cases. For the case where edit distance is known to be at most $k$, Ukkonen~\cite{ukkonen1985algorithms} provided an $O(nk)$ time algorithm and Landau et al.~\cite{landau1998incremental} improved upon that with an $O(n+k^2)$ time algorithm. For the case where the longest common subsequence (LCS) is known to be at most $L$, Hirschberg~\cite{Hirschberg77} gave an algorithm running in time $O(n\log n + Ln)$. Following a long line of works, Gawrychowski~\cite{Gawrychowski12} currently has the fastest algorithm for the special case of strings that can be compressed as small {\em straight-line programs} (SLP).
Andoni and Krauthgamer~\cite{AndoniK08} obtain efficient approximations to edit distance for the case where the strings are perturbed a-la smoothed analysis. Goldwasser and Holden~\cite{GoldwasserH17} obtain subquadratic algorithms when the input is augmented with auxiliary correlated strings.

Other special cases have also been considered (see also~\cite{BK18-LCS-multivariate}), but the work closest to ours is by Hunt and Szymanski~{\cite{hunt1977fast}}, who obtained a running time of $O((n+r)\log n)$ for the special case where there are $r$ ``matching pairs'', i.e.~pairs of identical characters (see also Section~\ref{sub:non-crossing}). 
While we directly build on their algorithm, note that there is an obstacle to applying it in our setting: for a constant size alphabet, we expect that a constant fraction of all $n^2$ pairs of characters will be matching, i.e. $r = \Theta(n^2)$.

\subsubsection*{Approximation algorithms}
There is a long line of works on efficient approximation algorithms for edit distance in the worst case~\cite{bar2004approximating,batu2006oblivious,andoni2010polylogarithmic,andoni2012approximating,boroujeni2018approximating,CDGKS18}.
First, it is important to note that even after the recent breakthrough of Chakraborty et al.~\cite{CDGKS18}, it is not known how to obtain approximation factors better than $3$ (see also discussion in~\cite{Rubinstein18-blog}). 
Furthermore, our running time is much faster than Chakraborty et al.'s~\cite{CDGKS18} and even faster than the near-linear time approximations of~\cite{andoni2010polylogarithmic,andoni2012approximating}. The best known approximation factor in time $O(n\polylog(n))$ is still worse than $n^{1/3}$~\cite{batu2006oblivious}.

In terms of techniques, our algorithm is most closely inspired by the window-compatible matching paradigm introduced by the recent quantum approximation algorithm of Boroujeni et al.~\cite{boroujeni2018approximating} (a similar idea was also used by Chakraborty et al.~\cite{CDGKS18}).  

\subsubsection*{A new ray of hope}
In this work we combine both approaches: 
namely we allow for (arbitrarily good) approximation, and also restrict our attention to the special case of computing the edit distance between a worst case input and a codeword from our code.
The interesting question thus becomes if there is a way to build enough structure into a string (or a set of strings/codewords) that allows for fast edit distance computations. Given the importance and pervasiveness of edit distance problems we find this to be a question of interest way beyond its applicability to synchronization codes. An independent work of Kuszmaul~\cite{KuszmaulPsuedorandomEDApprox19} also employs the combination of the two approaches and provides a near-linear time algorithm for approximating the edit distance between a pseudo-random string and an arbitrary one within a constant factor.

\section{Our Results}

In this paper, we introduce a simple and generic structure that achieves this goal. In particular, we will show that there exist strings over a finite alphabet that, if one indexes any given string $S$ with them, the edit distance of the resulting string to any other string $S'$ can be approximated within a $1+\eps$ factor in near-linear time. This also leads to breaking the quadratic decoding time barrier for insertion-deletion codes with near-optimal communication efficiency.

We start with a formal definition of \emph{string indexing} followed by the definition of an \emph{indexing scheme}.

\begin{definition}[String Indexing or Coordinate-Wise String Concatenation]
Let $S\in\Sigma^n$ and $S'\in\Sigma'^n$ be two strings of length $n$ over alphabets $\Sigma$ and $\Sigma'$. The \emph{coordinate-wise concatenation of $S$ and $S'$} or \emph{$S$ indexed by $S'$} is a string of length $n$ over alphabet $\Sigma\times\Sigma'$ whose $i$th element is $(S_i, S'_i)$. We denote this string with $S\times S'$.
\end{definition}

\begin{definition}[Indexing Scheme]
The pair $(I, \mathcal{\widetilde{\ED}}_I)$ consisting of string $I\in\Sigma_{\textnormal{Index}}^n$ and algorithm $\mathcal{\widetilde{\ED}}_I$ is an \emph{$\eps$-indexing scheme} if
for any string $S\in\Sigma^n$ and $S'\in[\Sigma\times\Sigma_{\textnormal{Index}}]^n$, $\mathcal{\widetilde{\ED}}_I(S\times I, S')$ outputs a set of up to 
$(1+\eps)\ED(S\times I, S')$
symbol insertions and symbol deletions over $S\times I$ that turns it into $S'$. The $\ED(\cdot)$ notation represents the edit distance function.
\end{definition}

The main result of this work is on the existence of indexing schemes that facilitate approximating the edit distance in near-linear time. 

\begin{theorem}\label{thm:near-linear-labeling}
For any $\eps\in(0, 1)$ and integer $n$, there exist a string $I\in\Sigma_{\textnormal{Index}}^n$ and an algorithm $\mathcal{\widetilde{\ED}}_I$ where $(I, \mathcal{\widetilde{\ED}}_I)$ form an $\eps$-indexing scheme, $|\Sigma_{\textnormal{Index}}|=\exp\left(\frac{\log (1/\eps)}{\eps^3}\right)$, $\mathcal{\widetilde{\ED}}_I$ runs in $O_\eps(n\poly(\log n))$ time, and $I$ can be constructed in $O_\eps(n\poly(\log n))$ time.
\end{theorem}

\subsection{Applications}
%
%
%

One application of indexing schemes that we introduce in this work is in enhancing the design of insertion-deletion codes (insdel codes) from~\cite{haeupler2017synchronization, haeupler2018synchronization4}.
The construction of codes from \cite{haeupler2017synchronization, haeupler2018synchronization4} consist of indexing each codeword of some appropriately chosen error correcting code with symbols of a synchronization string which, in the decoding procedure, will be used to recover the position of received symbols. As we will recapitulate in \cref{sec:applications}, this procedure of recovering the positions consists of several longest common subsequence computations between the utilized synchronization string and some other version of it that is altered by a number of insertions and deletions. This fundamental step resulted in an $\Omega(n^2)$ decoding time complexity for codes in~\cite{haeupler2017synchronization, haeupler2018synchronization4}.

Using the $\eps$-indexing schemes in this paper, we will modify constructions of~\cite{haeupler2017synchronization, haeupler2018synchronization4} so that the above-mentioned longest common subsequence computations can be replaced with approximations of the longest common subsequence (using \cref{thm:near-linear-labeling}) that run in near-linear time. The following theorem, that improves the main result of~\cite{haeupler2017synchronization} with respect to the decoding complexity, gives an insertion-deletion code for the entire range of distance that approaches the Singleton bound and is decodable in near-linear time.

\begin{theorem}\label{thm:coding_applications_unique}
For any $\eps>0$ and $\delta \in (0,1)$ there exists an encoding map $E: \Sigma^k \rightarrow \Sigma^n$ and a decoding map $D: \Sigma^* \rightarrow \Sigma^k$, such that, if $\ED(E(m),x) \leq \delta n$ then $D(x) = m$. Further, $\frac{k}{n} > 1 - \delta - \eps$, $|\Sigma|=\exp\left(\eps^{-4}\log (1/\eps)\right)$, and $E$ and $D$ are explicit and can be computed in linear and near-linear time in terms of $n$ respectively. 
\end{theorem}

A very similar improvement is also applicable to the design of list-decodable insertion-deletion codes from~\cite{haeupler2018synchronization4} as they also utilize indexed synchronization strings and a similar position recovery procedure along with an appropriately chosen list-recoverable code. (See \cref{def:ListRecoverableCodes})  In this case, we obtain list-decodable insertion-deletion codes that match the fastest known list-recoverable codes in terms of decoding time complexity.


\begin{theorem}\label{thm:list-decodable-codes}
For every $0<\delta,\eps<1$ and $\eps_0,\gamma > 0$, there exists a family of list-decodable codes that can protect against $\delta$-fraction of deletions and $\gamma$-fraction of insertions and achieves a rate of at least $1-\delta-\eps$ over an alphabet of size 
$O_{\eps_0, \eps, \gamma}\left(1\right)$.
There exists a randomized decoder for these codes with list size $L_{\eps_0, \eps, \gamma}(n)= \exp\left(\exp\left(\exp\left(\log^* n\right)\right)\right)$, $O(n^{1+\eps_0})$ encoding and decoding complexities, and decoding success probability 2/3.
\end{theorem}
Both \cref{thm:coding_applications_unique,thm:list-decodable-codes} are built upon the fact that if one indexes a synchronization string with an appropriate indexing scheme, the resulting string will be a synchronization string that is decodable in near-linear time.

\subsection{Other Results, Connection to List-Recovery, and Organization of the Paper}
In the rest of the paper, we first provide some preliminaries and useful lemmas from previous works in \cref{sec:prelim}. In \cref{sec:labeling}, we introduce the construction of our indexing schemes and prove \cref{thm:near-linear-labeling}. The construction of these indexing schemes utilize insertion-deletion codes that are list-decodable from large fractions of deletions and insertions. We use list-decodable codes from~\cite{haeupler2018synchronization4} for that purpose, which themselves use list-recoverable codes as a core building block. Therefore, the quality of indexing schemes that we provide, namely, time complexity and alphabet size, greatly depend on utilized list-recoverable codes and can be improved following the prospective advancement of list-recoverable codes in the future.

In \cref{sec:enhanced-labeling}, we enhance the structure of the indexing scheme from \cref{thm:near-linear-labeling} and provide \cref{thm:enhanced-labeling} that describes a construction of indexing schemes using $(\eps, \frac{1}{\eps}, L)$-list-recoverable codes as a black-box.
This result opens the door to potentially reduce the polylogarithmic terms in the time complexity of indexing schemes from \cref{thm:near-linear-labeling} by future developments in the design of list-recoverable codes. For instance, finding near-linear time $(\eps, \frac{1}{\eps}, \poly(\log n))$-list recoverable codes leads to indexing schemes that run in $O(n\poly(\log\log n))$ time via \cref{thm:enhanced-labeling}. 

As of the time of writing this paper, no such list-recoverable code is known. However, a recent work of Hemenway, Ron-Zewi, and Wootters~\cite{hemenway2017local} presents list-recoverable codes with $O\left(n^{1+\eps_0}\right)$ time probabilistic decoders for any $\eps_0>0$ that are appropriate for the purpose of being utilized in the construction of indexing schemes as outlined in \cref{thm:enhanced-labeling}. In \cref{sec:randomized-labeling}, we use such codes with the indexing scheme construction method of \cref{thm:enhanced-labeling} to provide a randomized indexing scheme with $O(n\log^{\eps_0} n)$ time complexity for any chosen $\eps_0>0$.

Then, in \cref{sec:applications}, we discuss the application of indexing schemes in the design of insertion-deletion codes. We start by \cref{lem:enhanced-sync-string-decoding} that enhances synchronization strings by using them along with indexing schemes and, therefore, enables us to reduce the time complexity of the position recovery subroutine of the decoders of codes  from~\cite{haeupler2017synchronization,haeupler2018synchronization4} to near-linear time. In \cref{sec:unique-codes}, we discuss our results for uniquely-decodable codes and prove \cref{thm:coding_applications_unique}. At the end, in \cref{sec:list-decodable-codes}, we address construction of list-decodable synchronization codes using indexing schemes. We start by \cref{thm:coding_applications} that gives a black-box conversion of a given list-recoverable code to a list-decodable insertion-deletion code by adding only a near-linear time overhead to the decoding complexity and, therefore, paves the path to obtaining insertion-deletion codes that are list-decodable in near-linear time upon the design of near-linear time list-recoverable codes. We use this conversion along with list-recoverable codes of~\cite{hemenway2017local} to prove~\cref{thm:list-decodable-codes}.

\section{Preliminaries and Notation}\label{sec:prelim}
In this section, we provide definitions and preliminaries that will be useful throughout the rest of the paper.

\subsection{Synchronization Strings}\label{sec:prelim-sync-strings}
We start by some essential definitions and lemmas regarding synchronization strings. Synchronization strings are defined as follows~\cite{haeupler2017synchronization}.

\begin{definition}[$\eps$-Synchronization Strings]\label{def:synchronization-strings}
String $S \in \Sigma^n$ is an $\eps$-synchronization string if for every $1 \leq i < j < k \leq n + 1$ we have that $\ED(S[i, j),S[j, k)) > (1-\eps) (k-i)$. 
\end{definition}

An important property of such strings is that they cannot have pairs of long identical disjoint subsequences.

\begin{theorem}[Theorem 6.2 of \cite{haeupler2017synchronization}]\label{lem:self-matching}
Let $S$ be an $\eps$-synchronization string of length $n$ and $1\le i_1<i_2<\cdots<i_l\le n$ and $1\le j_1<j_2<\cdots<j_l\le n$ be integers so that $S(i_k)=S(j_k)$ but $i_k \neq j_k$  for all $1\le k\le l$. Then $l\le \eps n$.
\end{theorem}

We will also make use of the following construction of synchronization strings that is developed in~\cite{haeupler2017synchronization3, cheng2018synchronization}.
\begin{theorem}[Theorem 1.3 from~\cite{cheng2018synchronization}]
For any $\eps\in(0,1)$ and integer $n$, one can construct an $\eps$-synchronization string of length $n$ over an alphabet of size $O\left(\eps^{-3}\right)$ in linear time.
\end{theorem}

Haeupler et al.~\cite{haeupler2018synchronization4} suggest a construction of list-decodable insertion-deletion codes by indexing the codewords of a list-recoverable code with symbols of a synchronization string. As we will use similar techniques and ideas throughout this paper, we formally define list-recoverable codes and review the main result of \cite{haeupler2018synchronization4} in the rest of this section.

\begin{definition}[List-recoverable codes]\label{def:ListRecoverableCodes}
Code $\mathcal{C}$ with encoding function $\Encode:\Sigma^{nr}\rightarrow\Sigma^n$ is called $(\alpha, l, L)$-list recoverable if for any collection of $n$ sets $S_1, S_2, \cdots, S_n\subset\Sigma$ of size $l$ or less, there are at most $L$ codewords for which more than $\alpha n$ elements appear in the list that corresponds to their position, i.e., $$\left|\left\{x \in \mathcal{C}\mid \left|\left\{i \in [n] \mid x_i \in S_i\right\}\right| \geq \allowbreak\alpha n \right\}\right| \leq L.$$
\end{definition}

\begin{theorem}[Theorem 1.1 from~\cite{haeupler2018synchronization4}]\label{thm:InsDelListDecoding}
For every $0<\delta,\eps<1$ and $\gamma > 0$, there exist a family of list-decodable insdel codes that can protect against $\delta$-fraction of deletions and $\gamma$-fraction of insertions and achieves a rate of $1-\delta-\eps$ or more over an alphabet of size 
$\left(\frac{\gamma+1}{\eps^2}\right)^{O\left(\frac{\gamma+1}{\eps^3}\right)}=O_{\gamma, \eps}\left(1\right)$. These codes are list-decodable with lists of size $L_{\eps, \gamma}(n)= \exp\left(\exp\left(\exp\left(\log^* n\right)\right)\right)$, and have polynomial time encoding and decoding complexities.
\end{theorem}

By choosing $\delta = \gamma = 1-\epsilon$ and $\eps=\epsilon/2$ in \cref{thm:InsDelListDecoding}, we derive the following corollary.

\begin{corollary}\label{cor:capacity_approaching_list_decodable_insdel}
For any $0<\eps<1$, there exists an alphabet $\Sigma_\eps$ with size $\exp(\eps^{-3}\log 1/\eps)$ and an infinite family of insertion-deletion codes, $\mathcal{C}$, that achieves a rate of $\frac{\eps}{2}$ and is  $L$-list-decodable from any $(1-\eps)n$ deletions and $(1-\eps)n$ insertions in polynomial time where $L=\exp(\exp(\exp(\log^* n)))$.
\end{corollary}

\subsection{Non-crossing Matchings}\label{sub:non-crossing}
The last element that we utilize as a preliminary in this paper is an algorithm provided in a work of Hunt and Szymanski~{\cite{hunt1977fast}} to compute the maximum \emph{non-crossing matching} in a bipartite graph. Let $G$ be a bipartite graph with an ordering for vertices in each part. A non-crossing matching in $G$ is a matching in which edges do not intersect. 

\begin{definition}[Non-Crossing Matching]
Let $G$ be a bipartite graph with ordered vertices $u_1, u_2, \cdots, u_m$ and $v_1, v_2, \cdots, v_n$ in each part. A non-crossing matching is a subset of edges of $G$ like 
$$\left\{(u_{i_1}, v_{j_1}), (u_{i_2}, v_{j_2}), \cdots, (u_{i_l}, v_{j_l})\right\}$$
 where $i_1<i_2<\cdots<i_l$ and $j_1<j_2<\cdots<j_l$.
\end{definition}

In this paper, we use an algorithm by Hunt and Szymanski~{\cite{hunt1977fast}} that essentially computes the largest non-crossing matching in a given bipartite graph.

\begin{theorem}[Theorem 2 of Hunt and Szymanski~{\cite{hunt1977fast}}] \label{thm:max_non_crossing_algorithm}
Let $G$ be a bipartite graph with $n$ ordered vertices in each part and $r$ edges. There is an algorithm that computes the largest non-crossing matching of $G$ in $O\left((n+r)\log\log n\right)$.
\end{theorem}
\section{Near-Linear Edit Distance Computations via Indexing}\label
{sec:labeling}
We start by a description of the string that will be used in our indexing scheme. Let $\mathcal{C}$ be an insertion-deletion code over alphabet $\Sigma_{\mathcal{C}}$, with block length $N$, and rate $r$ that is $L$-list decodable from any $N(1-\eps)$ deletions and $N(1-\eps)$ insertions in $T_{\Decode_{\mathcal{C}}}(N)$ for some sufficiently small $\eps > 0$.
We construct the indexing sequence $I$ by simply concatenating the codewords of $\mathcal{C}$. The construction of such indexing sequence resembles long-distance synchronization strings from~\cite{haeupler2017synchronization3}.

Throughout this section, we consider string $S$ of length $N\cdot\left|\Sigma_{\mathcal{C}}\right|^{Nr}$ that consists of coordinate-wise concatenation of a content string $m$ and the indexing string $I$. In other words, $S_i = (m_i, I_i)$. We will provide algorithms that approximate the edit distance of $S$ to a given string $S'$.

Consider the longest common subsequence between $S$ and $S'$. One can represent such common subsequence by a matching $\mathcal{M}_{LCS}$ with non-crossing edges in a bipartite graph with two parts of size $|S|$ and $|S'|$ where each vertex corresponds to a symbol in $S$ or $S'$ and each edge corresponds to a pair of identical symbols in the longest common subsequence.

Note that one can turn $S$ into $S'$ by simply deleting any symbol that corresponds to an unmatched vertex in $S$ and then inserting symbols that correspond to the unmatched vertices in $S'$. Therefore, the edit distance between $S$ and $S'$ is equal to the number of non-connected vertices in that graph.
To provide a $(1+\eps)$ edit distance approximation as described in \cref{thm:near-linear-labeling}, one only needs to compute a common subsequence, or equivalently, a non-crossing matching between $S$ and $S'$ in which the number of unmatched vertices does not exceed a $1+\eps$ multiplicative factor of $\mathcal{M}_{LCS}$'s.

We start by an informal intuitive justification of the algorithm. The algorithm starts by splitting the string $S'$ into blocks of length $N$ in the same spirit as $S$. We denote $i$th such block by $S'(i)$ and the $i$th block of $S$ by $S(i)$. Note that the blocks of $S$ are codewords of an insertion-deletion code with high distance indexed by $m$ ($S(i)=\mathcal{C}(i) \times m[N(i-1), Ni-1]$). Therefore, one might expect that any block of $S$ that is not significantly altered by insertions and deletions, (1) appears in a set of consecutive blocks in $S'$ and (2) has a small edit distance to at least one of those blocks.

Following this intuition, our proposed algorithm works thusly: For any block of $S'$ like $S'(i)$, the algorithm uses the list decoder of $\mathcal{C}$ to find all (up to $L$) blocks of $S$ that can be turned into $S'(i)$ by $N(1-\eps)$ deletions and $N(1-\eps)$ insertions ignoring the content portion on $S'$. In other words, let $S'(i)=C'_i\times m'[N(i-1), Ni-1]$. We denote the set of such blocks by $\Decode_{\mathcal{C}}(C'_i)$.
Then, the algorithm constructs a bipartite graph $G$ with $|S|$ and $|S'|$ vertices on each side (representing symbols of $S$ and $S'$) as follows: a symbol in $S'(i)$ is connected to all identical symbols in the blocks that appear in $\Decode_{\mathcal{C}}(C'_i)$ or any block that is in their 
\noSTOC{$w=O\left(\frac{1}{\eps}\right)$}
\STOConly{$w=O\left(1/\eps\right)$}
 neighborhood, i.e., is up to 
\noSTOC{$O\left(\frac{1}{\eps}\right)$}
\STOConly{$O\left(1/\eps\right)$}
  blocks away from at least one of the members of $\Decode_{\mathcal{C}}(C'_i)$.

Note that any non-crossing matching in $G$ corresponds to some common subsequence between $S$ and $S'$ because $G$'s edges only connect identical symbols. In the next step, the algorithm finds the largest non-crossing matching in $G$, $\mathcal{M}_{ALG}$, and outputs the corresponding set of insertions and deletions as the output. We will use the algorithm proposed by Hunt and Szymanski~\cite{hunt1977fast} (see \cref{thm:max_non_crossing_algorithm}) to find the largest non-crossing matching. A formal description of the algorithm is available in~\cref{alg:EditDistanceApprox}.

\begin{algorithm}
\caption{$(1+11\eps)$-Approximation for Edit Distance}\label{alg:EditDistanceApprox}
\begin{algorithmic}[1]
\Procedure{ED-Approx}{$S, S', N, \Decode_{\mathcal{C}}(\cdot)$}

\State Make empty bipartite graph $G$ with parts of size $(|S|, |S'|)$
\State $w=\frac{1}{\eps}$

\For{{\bf each} $S'(i)=C'_i\times m'[N(i-1), Ni-1]$}
\State{$List \leftarrow \Decode_{\mathcal{C}}(C'_i)$}
\For{{\bf each} $j \in List$}
\For{$k\in\left[j-w, j+w\right]$}
\State \label{step:adding-edges-to-G}Connect pairs of vertices in $G$ that correspond to identical symbols in $S(k)$ and $S'(i)$.
\EndFor
\EndFor
\EndFor

\State $\mathcal{M}_{ALG} \leftarrow$ Largest non-crossing matching in $G$ \noSTOC{(Using~\cref{thm:max_non_crossing_algorithm})}

\State \textbf{return} $\mathcal{M}_{ALG}$
\EndProcedure
\end{algorithmic}
\end{algorithm}

\subsection{Analysis}
We now proceed to the analysis of approximation guarantee and time complexity of~\cref{alg:EditDistanceApprox}.

\begin{theorem}\label{thm:edit-distance-approximation-time}
For $n=\max(|S|, |S'|)$, the running time of \cref{alg:EditDistanceApprox} is {$O\left(\frac{n}{N}\cdot T_{\Decode_{\mathcal{C}}}(N) + \frac{NL}{\eps}\cdot n\log\log n\right)$}.
\end{theorem}
\begin{proof}
The algorithm starts by using the decoder for any block in $S'$ which takes a total of $\frac{n}{N}\cdot T_{\Decode_{\mathcal{C}}}(N)$ time. Further, construction of $G$ will take \noSTOC{$O\left(nLN\frac{1}{\eps}\right)$}\STOConly{$O\left(nLN/\eps\right)$}.
$G$ has no more than 
\noSTOC{$n\cdot NL\cdot w= O\left(\frac{nNL}{\eps}\right)$}\STOConly{$n\cdot NL\cdot w= O\left(nNL/\eps\right)$}
edges.
Thus, by using Hunt and Szymanski's~\cite{hunt1977fast} algorithm (\cref{thm:max_non_crossing_algorithm}), the maximum non-crossing matching in $G$ can be computed in 
\noSTOC{$O\left((n+\frac{nNL}{\eps})\log\log n\right)=O\left(\frac{NL}{\eps}\cdot n\log\log n\right)$.}\STOConly{$O\left(\frac{NL}{\eps}\cdot n\log\log n\right)$.}
\end{proof}

Before providing the analysis for the approximation ratio of \cref{alg:EditDistanceApprox}, we define the following useful notions.
\begin{definition}[Projection]
Let $\mathcal{M}$ be a non-crossing matching between $S$ and $S'$. The \emph{projection of $S'(i)$ under $\mathcal{M}$} is defined to be the substring of $S$ between the leftmost and the rightmost element of $S$ that are connected to $S'(i)$ in $\mathcal{M}$. (see~\cref{fig:window-projection} for an example)
\end{definition}
\begin{definition}[Window Limited]
A non-crossing matching between $S$ and $S'$ is called \emph{$w$-window-limited} if the projection of any block of $S'$ fits in $w$ consecutive blocks of $S$. 
\end{definition}
The definition of window-limited matchings is inspired by the window-compatibility notion from~\cite{boroujeni2018approximating}.
\noSTOC{\begin{figure}
\centering
\includegraphics[width=.7\linewidth]{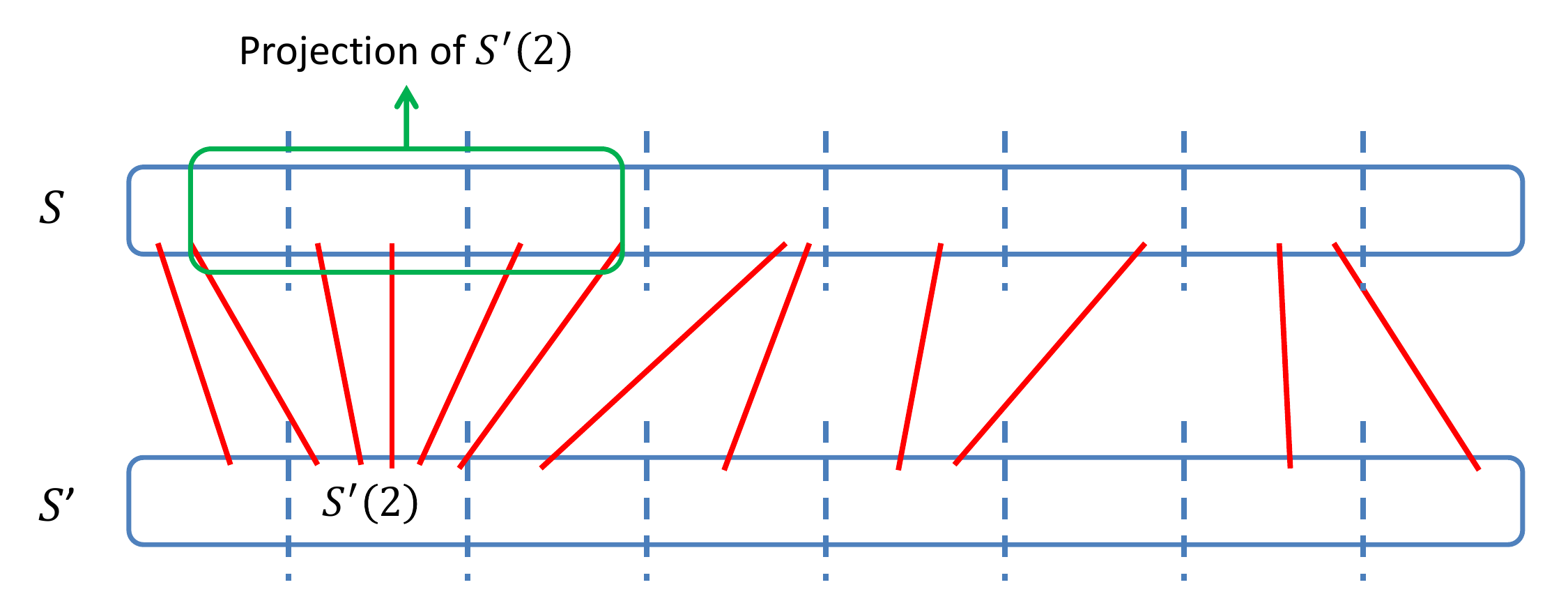}
\caption{An example of a matching between $S$ and $S'$ depicting the projection of $S'(2)$. This matching is 3-window-limited.}
\label{fig:window-projection}
\end{figure}}
\STOConly{\begin{figure}
\centering
\includegraphics[width=\linewidth]{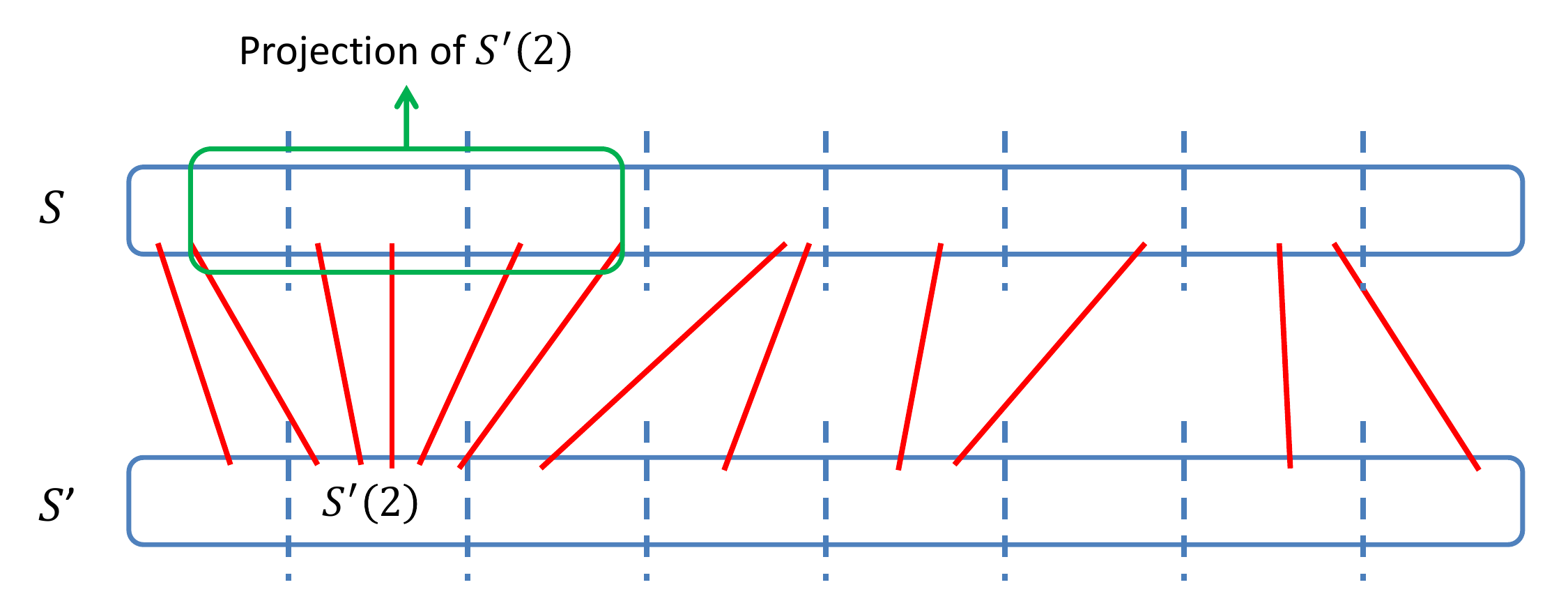}
\caption{An example of a matching between $S$ and $S'$ depicting the projection of $S'(2)$. This matching is 3-window-limited.}
\label{fig:window-projection}
\end{figure}}

\begin{theorem}\label{thm:edit-distance-approximation-ratio}
For $0<\eps<\frac{1}{21}$, \cref{alg:EditDistanceApprox} computes a set of up to $(1+11\eps)\cdot\ED(S, S')$ insertions and deletions that turn $S$ into $S'$.
\end{theorem}
\begin{proof}
Let $\ED_{ALG}$ denote the edit distance solution obtained by the matching suggested by~\cref{alg:EditDistanceApprox}.
We will prove that $\ED_{ALG} \le (1+11\eps)\cdot\ED(S, S')$ in the following two steps:
\begin{enumerate}
\item \label{item:first-step} 
Let $\mathcal{M}_W$ be the largest $w=\left(\frac{1}{\eps}+1\right)$-window-limited matching between $S$ and $S'$  and $\ED_W$ be its count of unmatched vertices. In the first step, we show the following.
\begin{equation}\label{eqn:step1}
\ED_W \le (1+3\eps)\ED(S, S')
\end{equation}

To prove this, consider $\mathcal{M}_{LCS}$, the matching that corresponds to the longest common subsequence. Then, we modify this matching by deleting all the edges connected to any block $S'(i)$ that violates the $w$-window-limited requirement. In other words, if the projection of $S'(i)$ spans over at least $w+1$ blocks in $S$, we remove all the edges with one endpoint in $S'(i)$. Note that removing the edges connected to $S'(i)$ might increase the number of unmatched vertices in the matching by $2N$. However, as projection of $S'(i)$ spans over at least $w+1$ blocks in $S$, one can assign all the originally unmatched vertices in that projection, which are at least $(w-1)\cdot N - N \ge (w-2)N$, to the newly introduced unmatched edges as an ``approximation budget''. Note that this assignment is mutually exclusive since projections of two distinct blocks of $S'$ are disjoint. Therefore, the above-mentioned removal procedure increases the number of unmatched vertices by a multiplicative factor no larger than 
$\frac{(w-2)N + 2N}{(w-2)N} = \frac{w}{w-2} = \frac{1+\eps}{1-\eps} \le 1+3\eps$ for $\eps\le\frac{1}{3}$.

Note that the matching obtained by the above-mentioned removal procedure is a $w$-window-limited matching and, therefore, has at least $\ED_W$ unmatched vertices by the definition of $\mathcal{M}_W$. Hence, \cref{eqn:step1} is proved.

\item \label{item:second-step} In the second step, we show that 
\begin{equation}\label{eqn:step2}
\ED_{ALG}\le (1+7\eps)\ED_W.
\end{equation}

Similar to Step \ref{item:first-step}, consider the largest $w$-window-limited matching $\mathcal{M}_{W}$ and then modify it by removing all the edges connected to any block $S'(i)$ that has less than $\eps N$ edges to any block in $S$. Again, we prove an approximation ratio by exclusively assigning some of the unmatched vertices in $\mathcal{M}_{W}$ to each $S'(i)$ that we choose to remove its edges.

Consider some $S'(i)$ that has less than $\eps N$ edges to any block in $S$. We assign all unmatched vertices in $S'(i)$ and all unmatched vertices in the projection of $S'(i)$ as the approximation budget for eliminated edges. Let $B$ be the number of blocks in $S$ that are contained or intersect with projection of $S'(i)$. As $S'(i)$ has less than $\eps N$ edges to any block in $S$, the total number of removed edges is less than $NB\eps$. This gives that there are at least $N-NB\eps$ unmatched vertices within $S'$ and $\max\{B-2,0\}\cdot N(1-\eps)$ unmatched vertices in its projection that are assigned to $2NB\eps$ new unmatched edges appearing as a result of removing $S'(i)$'s edges. Therefore, this process does not increase the number of unmatched vertices by a multiplicative factor more than
$
1+\frac{2NB\eps}{\left(N-NB\eps\right) + \max\{B-2,0\}\cdot N(1-\eps)}.$

If $B=1$ or 2, the above approximation ratio can be bounded above by $1+\frac{2NB\eps}{\left(N-NB\eps\right)} \le 1+\frac{4\eps}{1-2\eps} \le 1+5\eps$ for $\eps\le\frac{1}{10}$. Unless, $B\ge3$, therefore the approximation ratio is less than $1+\frac{2NB\eps}{(B-2)N(1-\eps)} \le 1+\frac{6\eps}{1-\eps} \le 1+7\eps$ for $\eps\le\frac{1}{7}$. Therefore, the edge removal process in Step~\ref{item:second-step} does not increase the number of unmatched vertices by a factor larger than $1+7\eps$.

Note that the matching obtained after the above-mentioned procedure is a $w$-window limited one in which any block of $S'$ that contains at least one edge, has more than $N\eps$ edges to some block in $S$ within its projection. Therefore, this matching is a subgraph of $G$. Since $\mathcal{M}_{ALG}$ is defined to be the largest non-crossing matching in $G$, the number of unmatched vertices in $\mathcal{M}_{ALG}$, $\ED_{ALG}$ is not larger than the ones in the matching we obtained in Step~\ref{item:second-step}. Hence, proof of \cref{eqn:step2} is complete.
\end{enumerate}

Combining \cref{eqn:step1,eqn:step2} implies the following approximation ratio.
\begin{eqnarray}
\ED_{ALG}\le (1+3\eps)(1+7\eps)\ED(S, S') \le (1+11\eps)\ED(S, S')
\end{eqnarray}
The last inequality holds for $\eps\le\frac{1}{21}$.
\end{proof}

\subsection{Proof of \cref{thm:near-linear-labeling}}
\begin{proof}
To prove this theorem, take $\eps'=\frac{\eps}{11}$. 
Further, take $\mathcal{C}$ as an insertion-deletion code from \cref{thm:InsDelListDecoding} with block length $N=c_0\cdot\frac{\log n \cdot \eps'^3}{(1-2\eps')\log(1/\eps')}$ and parameters $\delta_{\mathcal{C}}=\gamma_{\mathcal{C}}=1-\eps'$, $\eps_{\mathcal{C}}=\eps'$. (constant $c_0$ will be determined later)

According to \cref{thm:InsDelListDecoding}, $\mathcal{C}$ is $O_\eps(\exp(\exp(\exp(\log^*n))))$-list decodable from $(1-\eps')N$ insertions and $(1-\eps')N$ deletions, is over an alphabet of size $q_{\mathcal{C}}=\eps'^{-O(1/\eps'^{3})}=\exp\left(\frac{\log (1/\eps')}{\eps'^3}\right)$, and has rate $r_{\mathcal{C}}=1-2\eps'$.

Construct string $I$ according to the structure described in the beginning of \cref{sec:labeling} using $\mathcal{C}$ as the required list-decodable insertion-deletion code. Note that $|I|=N\cdot q_{\mathcal{C}}^{r_{\mathcal{C}}N} = N\cdot\exp\left(c_0\cdot O(\log n)\right)$. Choosing an appropriate constant $c_0$ that cancels out the constants hidden in $O$-notation that originate from hidden constants in the alphabet size will lead to $|I| = Nn = O(n\log n)$. Truncate the extra elements to have string $I$ of length $n$. As $\mathcal{C}$ is efficiently encodable, string $I$ can be constructed in near-linear time.

Further, define algorithm $\mathcal{\widetilde{\ED}}_I$ as follows. $\mathcal{\widetilde{\ED}}_I$ takes $S\times I$ and $S'$ and runs an instance of \cref{alg:EditDistanceApprox} with $S\times I$, $S'$, $N$, and the decoder of $\mathcal{C}$ as its input. \cref{thm:edit-distance-approximation-ratio} guarantees that $\mathcal{\widetilde{\ED}}_I(S\times I, S')$ generates a set of at most $(1+11\eps')\ED(S\times I, S') = (1+\eps)\ED(S\times I, S')$ insertions and deletions over $S\times I$ that converts it to $S'$. Finally, \cref{thm:edit-distance-approximation-time} guarantees that $\mathcal{\widetilde{\ED}}_I$ runs in 
\noSTOC{\begin{eqnarray*}
&&O\left(\frac{n}{N}\cdot T_{\Decode_{\mathcal{C}}}(N) + \frac{NL}{\eps}\cdot n\log\log n\right)\\ 
&=& O_\eps\left(\frac{n}{\log n}T_{\Decode_{\mathcal{C}}}(\log n)+n\log n \log\log n\exp(\exp(\exp(\log^*n)))\right)\\
&=&O_\eps(n\poly(\log n))
\end{eqnarray*}}
\STOConly{\begin{eqnarray*}
&&O\left(\frac{n}{N}\cdot T_{\Decode_{\mathcal{C}}}(N) + \frac{NL}{\eps}\cdot n\log\log n\right)\\ 
&=& O_\eps\left(\frac{nT_{\Decode_{\mathcal{C}}}(\log n)}{\log n}+n\log n \log\log n\exp(\exp(\exp(\log^*n)))\right)\\
&=&O_\eps(n\poly(\log n))
\end{eqnarray*}}
 time.
\end{proof}


\section{Enhanced Indexing Scheme}\label{sec:enhanced-labeling}
In \cref{sec:labeling}, we provided an indexing scheme, using which, one can essentially approximate the edit distance by a $(1+\eps)$ multiplicative factor for any $\eps > 0$. Note that if code $\mathcal{C}$ that was used in that construction has some constant rate $r=O_{\eps}(1)$, then $|S|=N\cdot |\Sigma_{\mathcal{C}}|^{Nr}$ and, therefore, 
\noSTOC{$N=\Theta_\eps\left(\frac{\log n}{r}\right)$. }
\STOConly{$N=\Theta_\eps\left(\log n / r\right)$. }
This makes the running time of~\cref{alg:EditDistanceApprox} from~\cref{thm:edit-distance-approximation-time} 
\noSTOC{$O\left(\frac{nr}{\log n}\cdot T_{\Decode_{\mathcal{C}}}(\log n) + \frac{\log n\cdot L}{\eps}\cdot n\log\log n\right)$.}
\STOConly{$O(nr / \log n\cdot T_{\Decode_{\mathcal{C}}}(\log n) \allowbreak+ \log n\cdot L/\eps\cdot n\log\log n)$.}
As described in the proof of \cref{thm:near-linear-labeling}, using the efficient list-decodable codes from \cref{cor:capacity_approaching_list_decodable_insdel}, one can obtain edit distance computations in $O_\eps(n\cdot\text{poly}(\log n) + n\log n\cdot\log\log n\cdot\allowbreak\exp\left(\exp(\exp(\log^* n))\right)) = O_\eps(n\cdot\text{poly}(\log n))$. 

In this section, we try to enhance this running time by reducing the poly-logarithmic terms. To this end, we break down the factors in our construction and edit distance computation that contribute to the poly-logarithmic terms in the decoding time complexity.
\begin{enumerate}
\item \textbf{Edges in graph $G$}: The number of edges in graph $G$ can be as high as 
\noSTOC{$\Theta\left(\frac{nNL}{\eps}\right)=\Theta(n\log n\cdot\text{poly}(\log\log n))$}\STOConly{$\Theta\left(nNL/\eps\right)=\Theta(n\log n\cdot\text{poly}(\log\log n))$}
 which, as discussed above, leads to an additive $n\log n\cdot\log\log n\cdot\exp\left(\exp(\exp(\log^* n))\right)$ component. In \cref{sec:two-layers}, we will show that this component can be reduced to $O(n\cdot\text{poly}(\log\log n))$ by having two layers of indices via indexing each codeword of $\mathcal{C}$ with an indexing scheme as described in \cref{sec:labeling} (constructed based on some code of block length $O(\log \log n)$).

\item \textbf{Decoding complexity of \noSTOC{code }$\mathcal{C}$ from~\cref{cor:capacity_approaching_list_decodable_insdel} ($T_{\Decode_{\mathcal{C}}}(\cdot)$)}:
As described in~\cref{sec:prelim-sync-strings}, list-decodable insdel codes from~\cref{thm:InsDelListDecoding} are obtained by indexing codewords of a list-recoverable code with a synchronization string and their decoding procedure consist of (1) calculating a constant number of longest common subsequence computations, and (2) running the decoder of the list-recoverable code.

Part (1) consumes quadratic time in terms of $N$. However, using the indexing schemes for approximating edit distance from \cref{thm:near-linear-labeling}, we will show in \cref{thm:coding_applications} that one can reduce the running time of part (1) to 
\noSTOC{$O\left(\frac{n}{\log n}\cdot \log n \cdot \text{poly}(\log\log n)\right) = O\left(n\cdot\text{poly}(\log\log n)\right)$.}\STOConly{$O\Big(\frac{n}{\log n}\cdot \log n \cdot \text{poly}(\log\log n)\Big) = O\left(n\cdot\text{poly}(\log\log n)\right)$.}

\end{enumerate}

Applying the above-metioned enhancements to the structure of our indexing scheme will result in the black-box construction of indexing schemes using list-recoverable codes as formalized in the following theorem. 

\begin{theorem}\label{thm:enhanced-labeling}
For any $\eps\in(0,1)$, given a family of codes over alphabet $\Sigma$ that are $\left(\frac{\eps}{46}, \frac{276}{\eps}, L(\cdot)\right)$-list recoverable in 
$T_{\Decode}(\cdot)$ time and achieve a rate of $r=O_\eps(1)$, one can construct an $\eps$-indexing scheme $(I, \mathcal{\widetilde{\ED}}_I)$ with any positive length $n$ over an alphabet of size $|\Sigma|^2\times \exp\left(\frac{\log (1/\eps)}{\eps^3}\right)$ where $\mathcal{\widetilde{\ED}}_I$ has 
\noSTOC{$$O_\eps\left(n\cdot\left[\frac{T_{\Decode}(\log n)}{\log n} + \frac{T_{\Decode}(\log \log n)}{\log\log n}+ \log^2\log n \cdot L(\log n)L(\log\log n) + \poly(\log\log n)	\right]\right)$$}\STOConly{$O_\eps\bigg(n\cdot\bigg[\frac{T_{\Decode}(\log n)}{\log n} + \frac{T_{\Decode}(\log \log n)}{\log\log n}+ \log^2\log n \cdot L(\log n)L(\log\log n) + \poly(\log\log n)	\bigg]\bigg)$}
running time complexity. Further, if the given family of codes are efficiently encodable, $I$ can be constructed in near-linear time.
\end{theorem}

These enhancements do not eventually yield an indexing scheme that works in $O(n\cdot\text{poly}(\log \log n))$ as the bottleneck of the indexing scheme's time complexity is the decoding time of the utilized list-recoverable code. 

As of the time of writing this paper, no deterministic list recoverable code with our ideal properties and a decoding time complexity faster than an unspecified large polynomial is found. However, because of the enhancements discussed in this section, improvements in decoding time complexity of list-recoverable codes can lead to $\eps$-indexing schemes that run in $O(n\cdot\text{poly}(\log\log n))$ time. Particularly, having a linear-time $\left(\eps,1/\eps, L(n)=\poly(\log n)\right)$-list recoverable code would suffice.

\subsection{Two Layer Indexing}\label{sec:two-layers}
Our enhanced indexing sequence $I$ consists of the coordinate-wise concatenation of two string $I_1$ and $I_2$ where $I_1$ is the ordinary indexing sequence as described in \cref{sec:labeling}, i.e, the codewords of a code $\mathcal{C}_1$ with block length $N_1$, and $I_2$ is repetitions of an ordinary indexing sequence $I'$ of length $N_1$ constructed using some code $\mathcal{C}_2$. (See \cref{fig:enhanced-construction})

\noSTOC{\begin{figure} 
    \begin{subfigure}[b]{\linewidth}
	\centering
     \includegraphics[width=120mm]{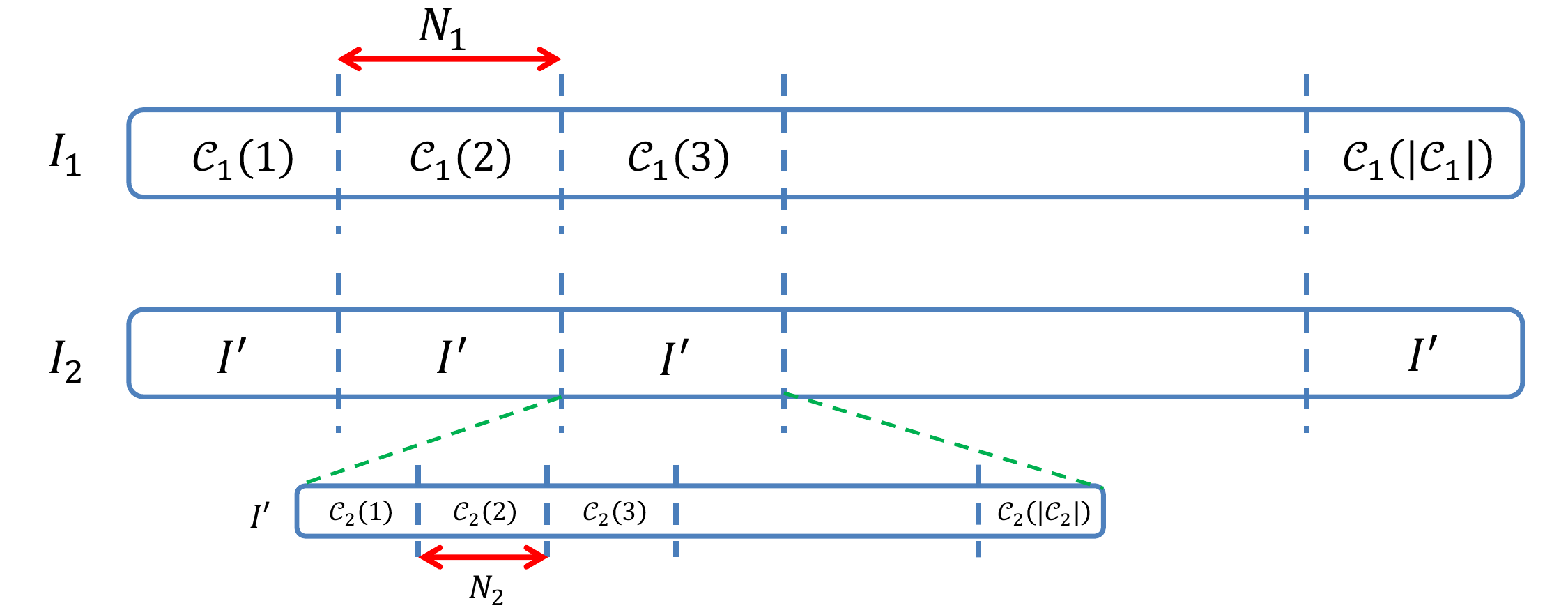}
     \caption{Construction of enhanced indexing string.}
     \label{fig:enhanced-construction}
    \end{subfigure} %

    \begin{subfigure}[b]{\linewidth}    
	\centering
     \includegraphics[width=120mm]{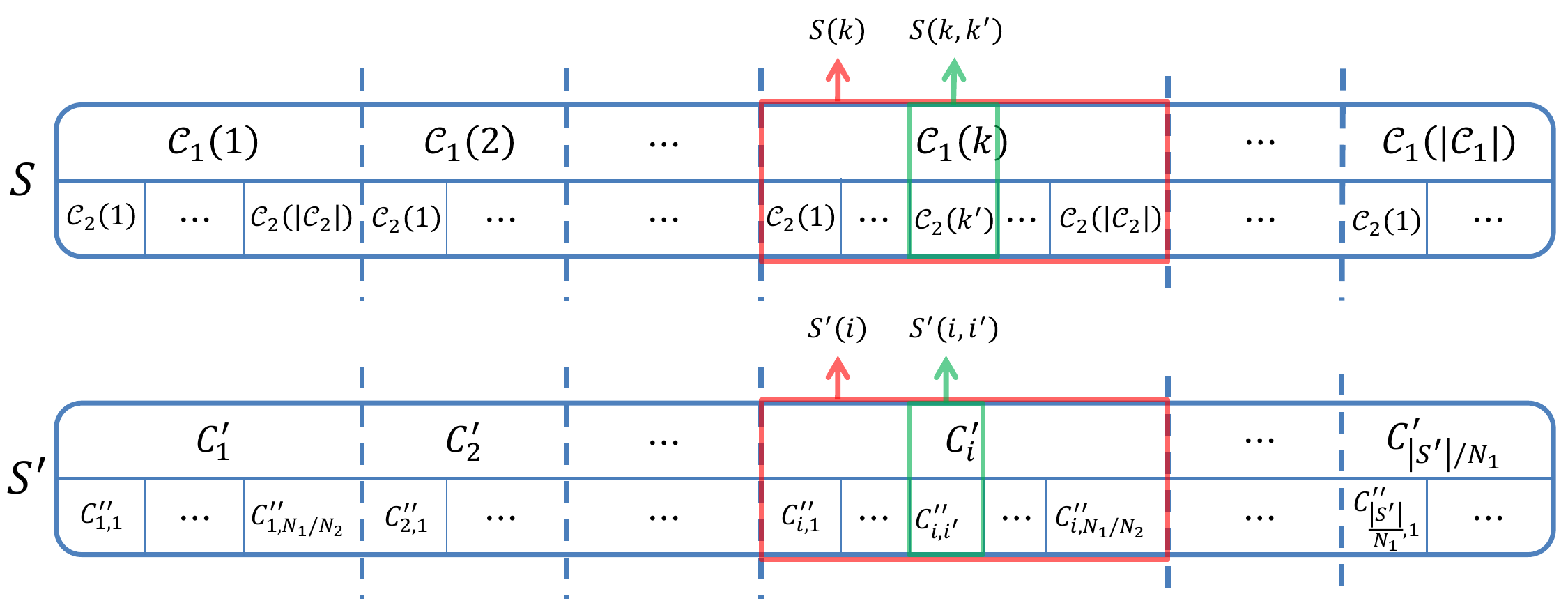}
	 \caption{Decoding for enhanced construction.}
      \label{fig:enhanced-decoding}
    \end{subfigure} 
\caption{}
\end{figure}}
\STOConly{\begin{figure} 
    \begin{subfigure}[b]{\linewidth}
	\centering
     \includegraphics[width=\linewidth]{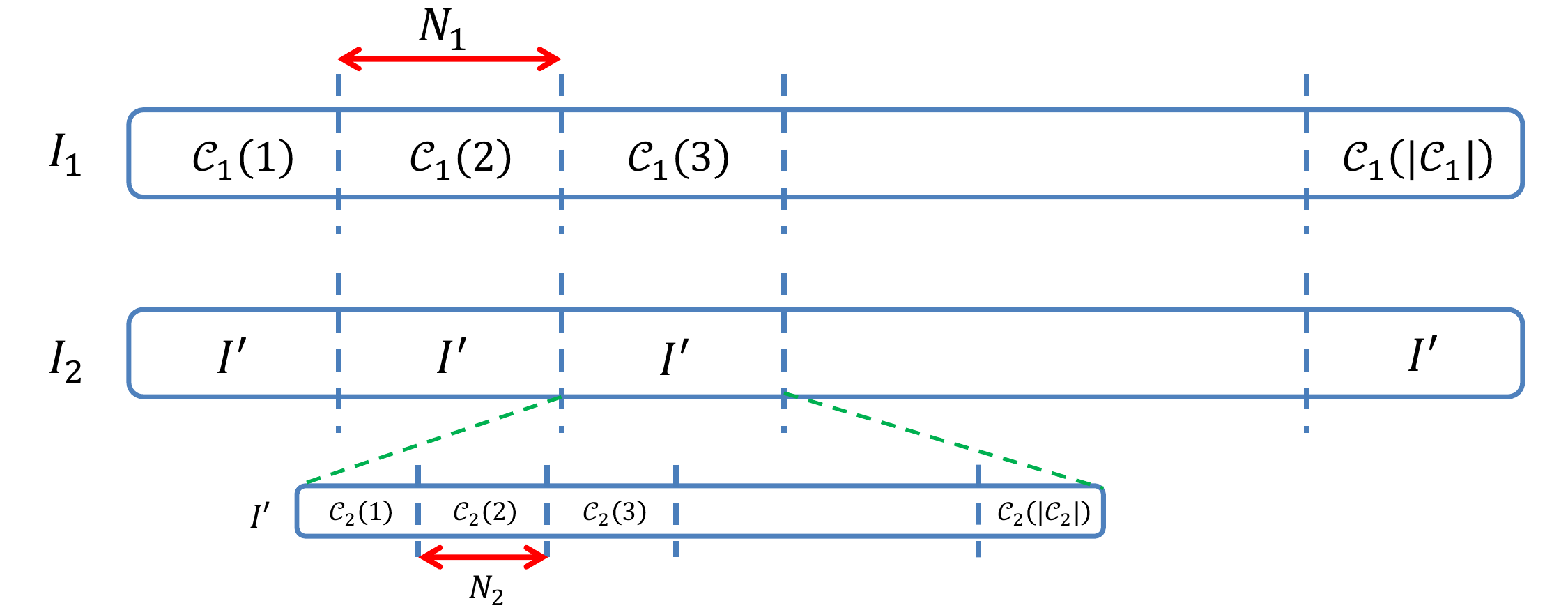}
     \caption{Construction of enhanced indexing string.}
     \label{fig:enhanced-construction}
    \end{subfigure} %

    \begin{subfigure}[b]{\linewidth}    
	\centering
     \includegraphics[width=\linewidth]{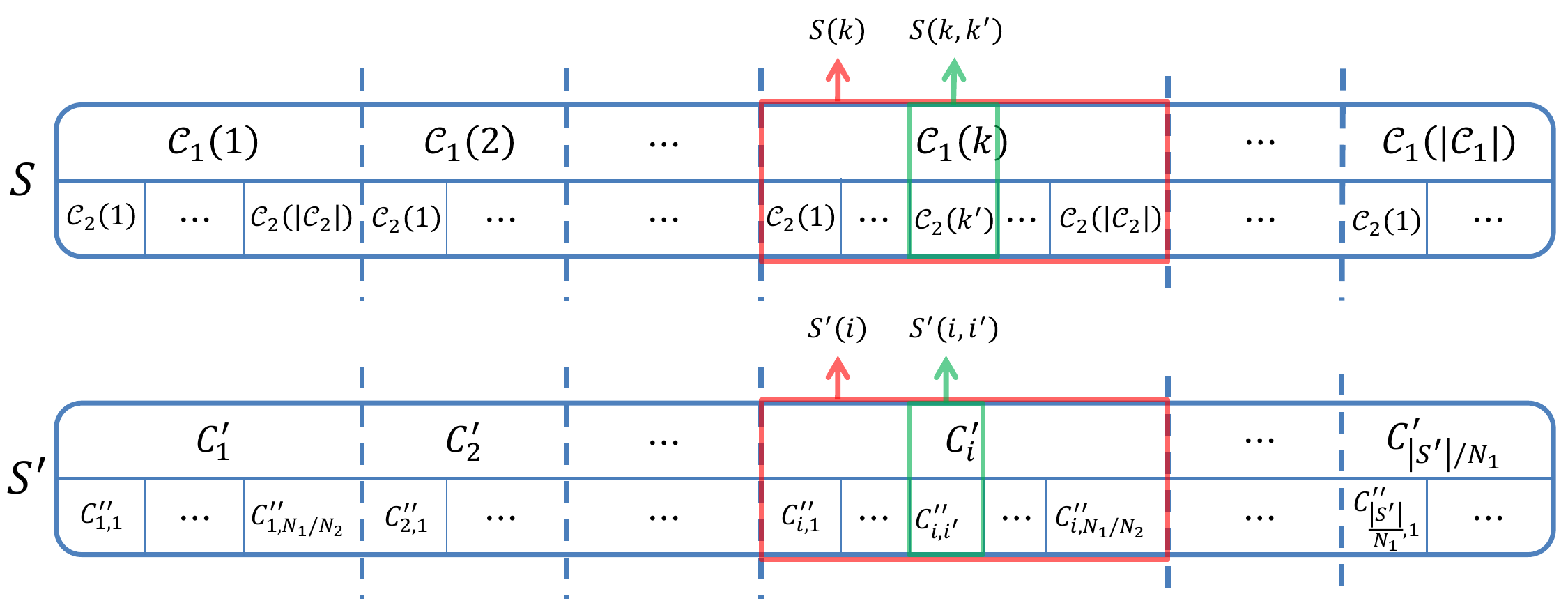}
	 \caption{Decoding for enhanced construction.}
      \label{fig:enhanced-decoding}
    \end{subfigure} 
\caption{}
\end{figure}}

In other words, let $\mathcal{C}_2$ be a code of block length $N_2$ and rate $r_2$ over alphabet $\Sigma_{\mathcal{C}_2}$ that is $L_2$-list decodable from $N_2(1-\eps)$ insertions and $N_2(1-\eps)$ deletions. Writing the codewords of $\mathcal{C}_2$ back to back would give the string $I'$ of length $|I'|=N_2\cdot |\Sigma_{\mathcal{C}_2}|^{N_2r_2}$. Then, let code $\mathcal{C}_1$ be a code of block length $N_1 = |I'|$ and rate $r_1$ over alphabet $\Sigma_{\mathcal{C}_1}$ that is $L_1$-list decodable from $N_1(1-\eps)$ insertions and $N_1(1-\eps)$ deletions. We form string $I_1$ by writing the codewords of $\mathcal{C}_1$ one after another and string $I_2$ by repeating $I'$ for $|\mathcal{C}_1|$ times. Finally, $I=(I_1, I_2)$.

We provide a decoding algorithm for indexing sequence $I$ that is very similar to \cref{alg:EditDistanceApprox} with an extra step in the construction of bipartite graph $G$ that reduces the number of edges at the cost of a weaker yet still constant approximation guarantee.

In Line \ref{step:adding-edges-to-G} of \cref{alg:EditDistanceApprox}, instead of adding an edge between any two pair of identical symbols in $S(k)$ and $S'(i)$ (that can be as many as $\log^2 n$), the algorithm runs another level of list-decoding and window-limiting based on the copy of $I'$ that is a component of $S(k)$. In other word, the algorithm uses the decoder of $\mathcal{C}_2$ for any sub-block of length $N_2$ in $S'(i)$, like $S'(i,i')$, to find up to $L_2$ sub-blocks of length $N_2$ in $S(k)$, like $S(k, k')$, and adds an edge between any two identical symbols between $S'(i,i')$ and $S(k,k')$. We denote the portion of $S'(i,i')$ that corresponds to $\mathcal{C}_2$ codewords by $C''_{i,i'}$. (See \cref{fig:enhanced-decoding}) A formal description is available in \cref{alg:EnhancedEditDistanceApprox}.

\begin{algorithm}
\caption{$(1+23\eps)$-Approximation for Edit Distance}\label{alg:EnhancedEditDistanceApprox}
\begin{algorithmic}[1]
\noSTOC{\Procedure{Enhanced-ED-Approx}{$S, S', N_1, N_2, \Decode_{\mathcal{C}_1}(\cdot), \Decode_{\mathcal{C}_2}(\cdot)$}}
\STOConly{\Procedure{Enhanced-ED-Approx}{$S, S', \{N_i, \Decode_{\mathcal{C}_i}(\cdot)\}_{i=1}^2$}}

\State Make empty bipartite graph $G$ with parts of size $|S|$ and $|S'|$
\State $w=\frac{1}{\eps}$

\For{{\bf each} $S'(i)=C'_{i}\times \left[C''_{i,1}, C''_{i,2}, \cdots, C''_{i,{N_1/N_2}}\right] \times m'[N_1(i-1), N_1i-1]$}
\State {$List_1 \leftarrow \Decode_{\mathcal{C}_1}(C'_i)$}
\For{{\bf each} $j \in List_1$}
\For{$k\in\left[j-w, j+w\right]$}
\For {$i' \in [1, N_1/N_2]$}
\State {$List_2 \leftarrow \Decode_{\mathcal{C}_2}(C''_{i, i'})$}
\For{{\bf each} $j' \in List_2$}
\For{$k'\in\left[j'-w, j'+w\right]$}
\State {Connect any pair of vertices in $G$ that correspond to identical symbols in $S(k,k')$ and $S'(i,i')$.}
\EndFor
\EndFor
\EndFor
\EndFor
\EndFor
\EndFor

\State $\mathcal{M}_{ALG} \leftarrow$ Largest non-crossing matching in $G$ \noSTOC{(Using~\cref{thm:max_non_crossing_algorithm})}

\State \textbf{return} $\mathcal{M}_{ALG}$
\EndProcedure
\end{algorithmic}
\end{algorithm}

\begin{theorem}\label{thm:enhanced-edit-distance-approximation-time}
\noSTOC{\cref{alg:EnhancedEditDistanceApprox} runs in  $O\left(\frac{n}{N_1}\cdot T_{\Decode_{\mathcal{C}_1}}(N_1) + \frac{n}{N_2}\cdot T_{\Decode_{\mathcal{C}_2}}(N_2) + \frac{N_2L_1L_2}{\eps^2}\cdot n\log\log n\right)$ time for $n=\max(|S|, |S'|)$.}
\STOConly{For $n=\max\left(|S|, |S'|\right)$, \cref{alg:EnhancedEditDistanceApprox} runs in $O\Big(\frac{n}{N_1}\cdot T_{\Decode_{\mathcal{C}_1}}(N_1) + \frac{n}{N_2}\cdot T_{\Decode_{\mathcal{C}_2}}(N_2) + \frac{N_2L_1L_2}{\eps^2}\cdot n\log\log n\Big)$ time.}
\end{theorem}
\begin{proof}
The algorithm uses the decoder of $\mathcal{C}_1$, $\frac{n}{N_1}$ times and the decoder of $\mathcal{C}_2$, $\frac{n}{N_2}$ times.
$G$ can have up to $\frac{n}{N}\cdot \frac{L_1}{\eps} \cdot \frac{N_1}{N_2} \cdot \frac{L_2}{\eps}\cdot N_2^2 = \frac{N_2 L_1 L_2}{\eps^2}\cdot n$ edges. Therefore, the use of Hunt and Szymanski's~\cite{hunt1977fast} algorithm (\cref{thm:max_non_crossing_algorithm}) will take $O\left(N_2L_1L_2/\eps^2\cdot n\log\log n\right)$ time. Therefore, the time complexity is as claimed.
\end{proof}

\begin{theorem}\label{thm:enhanced-edit-distance-approximation-ratio}
For $0<\eps<\frac{1}{121}$, \cref{alg:EnhancedEditDistanceApprox} computes a set of up to $(1+23\eps)\cdot\ED(S, S')$ insertions and deletions that turn $S$ into $S'$.
\end{theorem}
\begin{proof}
In the proof of~\cref{thm:edit-distance-approximation-ratio}, we proved that for the graph $G$ in \cref{alg:EditDistanceApprox},
$\ED_{ALG}\le (1+11\eps)\ED(S, S')$. In other words, the number of unmatched vertices in the largest non-crossing matching in that graph is at most $(1+11\eps)$ times the number of unmatched vertices in the bipartite graph that corresponds to the longest common subsequence between $S$ and $S'$.

As graph $G$ in \cref{alg:EnhancedEditDistanceApprox} is the same as the one in \cref{alg:EditDistanceApprox} with some extra edges removed, we only need to show that removing the extra edges does not  increase the number of non-matched vertices in the largest non-crossing matching by more than a $(1+O(\eps))$ multiplicative factor.
This can be directly concluded from \cref{thm:edit-distance-approximation-ratio} since the extra removed edges are eliminated by doing the same procedure between pairs of codewords of $\mathcal{C}_2$  that is done between the strings in the statement of \cref{thm:edit-distance-approximation-ratio}. In fact, using similar budget-based arguments as in \cref{eqn:step1,eqn:step2}, the extra edge removal step will only increase the edit distance by a $(1+11\eps)$ factor.
\noSTOC{This leads to the following upper bound on the approximation ratio of \cref{alg:EnhancedEditDistanceApprox} that holds for $\eps<\frac{1}{121}$.
$$(1+11\eps)(1+11\eps)\ED(S, S')\le (1+23\eps)\ED(S, S')$$}
\STOConly{This leads to an upper bound of 
$(1+11\eps)(1+11\eps)\ED(S, S')\le (1+23\eps)\ED(S, S')$
on the approximation ratio of \cref{alg:EnhancedEditDistanceApprox} for $\eps<\frac{1}{121}$.}
\end{proof}

\subsection{Proof of \cref{thm:enhanced-labeling}}
\begin{proof}
Let $\eps'=\eps/46$. Thus, the given family of codes is \STOConly{$\left(\eps', 6/\eps', L(\cdot)\right)$}\noSTOC{$\left(\eps', \frac{6}{\eps'}, L(\cdot)\right)$}-list recoverable.

Take the code $\mathcal{C}_1$ as a code with block length $N_1$ from the given family of codes where $N_1$ is large enough so that $N_1\cdot|\Sigma|^{r/2\cdot N_1} \geq n$. Similarly, take $\mathcal{C}_2$ with block length $N_2$ so that $N_2\cdot|\Sigma|^{r/2\cdot N_2} \geq N_1$. For a large enough $n$, rates of $\mathcal{C}_1$ and $\mathcal{C}_2$ are at least $r/2$. We reduce the rates of $\mathcal{C}_1$ and $\mathcal{C}_2$ to $r/2$ by arbitrarily removing codewords from them.

We now use \cref{thm:coding_applications} with parameters 
$\eps_{\text{conv}}=\eps'$ and $\gamma_{\text{conv}}=1-2\eps'$
to convert list-recoverable codes $\mathcal{C}_1$ and $\mathcal{C}_2$  to list-decodable insertion-deletion codes $\tilde{\mathcal{C}_1}$ and $\tilde{\mathcal{C}_2}$ by indexing their codewords with appropriately chosen indexing sequences from \cref{thm:near-linear-labeling} and synchronization strings. Note that we can do this conversion using \cref{thm:coding_applications} since $\gamma_{\text{conv}} = 1-2\eps' \le \frac{l_{\mathcal{C}_i} \cdot \eps_{\text{conv}}}{3}-1
= \frac{6/\eps' \cdot \eps'}{3}-1=1$. Also, $\tilde{\mathcal{C}_i}$ can $L(N_i)$-list decode from any $\gamma_{\text{conv}}=1-2\eps'$ fraction of insertions and any $1-\alpha_{\mathcal{C}_i}-\eps_{\text{conv}} = 1-\frac{\eps}{46}-\eps' = 1-2\eps'$ fraction of deletions in $T_{\Decode}(N_i)+O\left(N_i\poly(\log N_i)\right)$.

Also, it is known how to construct $\eps_s$-synchronization strings and $\eps_I$-indexing schemes needed in \cref{thm:coding_applications}. $\eps_s$-synchronization strings can be constructed in linear time in terms of their length over an alphabet of size $\eps_s^{-O(1)}$ and $\eps_I$-indexing sequences from \cref{thm:near-linear-labeling} can be constructed in near-linear time over an alphabet of size \noSTOC{$\exp\left(\frac{\log (1/\eps_I)}{\eps_I^3}\right)$}\STOConly{$\exp\left(\log (1/\eps_I) / \eps_I^3\right)$}. Therefore, the alphabets of $\tilde{\mathcal{C}_1}$ and $\tilde{\mathcal{C}_2}$ will be of size \noSTOC{$|\Sigma|\times\exp\left(\frac{\log (1/\eps')}{\eps'^3}\right)$}\STOConly{$|\Sigma|\times\exp\left(\log (1/\eps') / \eps'^3\right)$}. 

We now use codes $\tilde{\mathcal{C}_1}$ and $\tilde{\mathcal{C}_2}$ in the structure described in the beginning of \cref{sec:two-layers} to obtain an indexing sequence $I$ of length $n$. Since the conversion of each codeword of $\mathcal{C}_i$ to $\tilde{\mathcal{C}_i}$ consumes near-linear time in  terms of $N_i$, if the codes $\mathcal{C}_i$ are efficiently encodable, string $I$ can be constructed in near-linear time. Also, the above-mentioned discussion on alphabet sizes of $\tilde{\mathcal{C}_i}$ entails that $I$ will be a string over an alphabet of size 
\noSTOC{$|\Sigma|^2\times\exp\left(\frac{\log (1/\eps')}{\eps'^3}\right)$.}
\STOConly{$|\Sigma|^2\times\exp\left(\log (1/\eps') / \eps'^3\right)$.}

We now have to provide an algorithm that produces a $(1+\eps)$-approximation for the edit distance using $I$. In the same spirit as the algorithm provided in the proof of \cref{thm:enhanced-labeling}, we define algorithm $\mathcal{\widetilde{\ED}}_I$ as an algorithm that takes $S\times I$ and $S'$ and runs an instance of \cref{alg:EnhancedEditDistanceApprox} with $S\times l$, $S'$, $N_1$, $N_2$, and decoders of $\tilde{\mathcal{C}_i}$ as its input. 

As codes $\tilde{\mathcal{C}_i}$ list decode from $1-2\eps'$ fraction of insertions and deletions, \cref{thm:enhanced-edit-distance-approximation-ratio} guarantees that $\mathcal{\widetilde{\ED}}_I$ generates a set of at most $(1+23\cdot2(\eps'))\ED(S\times I, S') = (1+\eps)\ED(S\times I, S')$ insertions and deletions over $S\times I$ that converts it to $S'$. 

Finally, since $N_1 = O(\log n)$, $N_2=O(\log\log n)$ and $\tilde{\mathcal{C}_i}$ list decode in $T_{\Decode}(N_i)+O\left(N_i\poly(\log N_i)\right)$ time, \cref{thm:enhanced-edit-distance-approximation-time} guarantees that $\mathcal{\widetilde{\ED}}_I$ runs in 
\noSTOC{\begin{eqnarray*}
&&O\left(\frac{n}{N_1}\cdot T_{\Decode_{\tilde{\mathcal{C}_1}}}(N_1) + \frac{n}{N_2}\cdot T_{\Decode_{\tilde{\mathcal{C}_2}}}(N_2) + \frac{N_2L_1L_2}{\eps^2}\cdot n\log\log n\right)\\ 
&=& O_\eps\bigg(\frac{n}{\log n}\cdot\left[T_{\Decode}(\log n)+\log n\poly(\log \log n)\right]+\\
&&\qquad\frac{n}{\log\log n}\cdot\left[T_{\Decode}(\log \log n)+\log \log n\poly(\log \log \log n)\right]+\\
&&\qquad n\log^2\log n \cdot L(\log n) L(\log \log n)
\bigg)\\
&=&O_\eps\left(n\cdot\left[\frac{T_{\Decode}(\log n)}{\log n} + \frac{T_{\Decode}(\log \log n)}{\log\log n}+ \log^2\log n \cdot L(\log n)L(\log\log n) + \poly(\log\log n)	\right]\right)
\end{eqnarray*}
 time.}
\STOConly{
$$O\left(\frac{n}{N_1}\cdot T_{\Decode_{\tilde{\mathcal{C}_1}}}(N_1) + \frac{n}{N_2}\cdot T_{\Decode_{\tilde{\mathcal{C}_2}}}(N_2) + \frac{N_2L_1L_2}{\eps^2}\cdot n\log\log n\right)$$
 time. This gives the running time in the theorem statement.}
\end{proof}

\section{Randomized Indexing}\label{sec:randomized-labeling}
In this section, we will prove the following theorem by taking similar steps as in the proof of \cref{thm:enhanced-labeling} to construct an indexing scheme according to the structure introduced in \cref{sec:two-layers}.

\begin{theorem}\label{thm:randomized-labeling}
For any $\eps_0>0$, $\eps_1,\eps_2\in(0,1)$, and integer $n$, there exists a randomized indexing scheme $(I, \mathcal{\widetilde{\ED}}_I)$ of length $n$ where $\mathcal{\widetilde{\ED}}_I(S\times I, S')$ runs in $O(n\log^{\eps_0} n)$ time and proposes a set of insertions and deletions that turns $S\times I$ into $S'$ and contains up to $(1+\eps_1)\ED(S\times I, S') + \eps_2 |S'|$ operations with probability $1-\frac{1}{n^{O(1)}}$.
\end{theorem}

Note that, as opposed to the rest of the results in this paper, \cref{thm:randomized-labeling} provides an approximation guarantee with both multiplicative and additive components.

To construct such an indexing scheme using the structure introduced in \cref{sec:two-layers}, we will use a list-decodable insertion-deletion code of block length $O(\log\log n)$ from \cref{cor:capacity_approaching_list_decodable_insdel} and use \cref{thm:coding_applications} to obtain a list-decodable insertion-deletion code of block length $O(\log n)$ from the following list recoverable codes of~\cite{hemenway2017local}.

\begin{theorem}[Corollary of Theorem 7.1 of Hemenway et al.~\cite{hemenway2017local}]\label{thm:HemenwayEtAl}
For any $\rho\in[0, 1]$, $\eps>0$, and positive integer $l$, there exist constants $q_0$ and $c_0$ so that, for any $c<c_0$ and infinitely many integers $q\ge q_0$, there exists an infinite family of codes achieving the rate $\rho$ over an alphabet $\Sigma$ of size $|\Sigma|=q$ that is encodable in $n^{1+c}$ time and probabilistically $(\rho+\eps, l, L(n))$-list recoverable in $n^{1+c}$ time with success probability 2/3 and $L(n)=O_{\eps, \rho}(\exp(\exp(\exp(\log^* n))))$ where $n$ denotes the block length.
\end{theorem}

Before providing the proof of \cref{thm:randomized-labeling}, we mention a couple of necessary lemmas.

\begin{lemma}\label{lem:probability-enhancement}
Let $(\alpha, l, L(n))$-list-recoverable code $\mathcal{C}$ have a probabilistic decoder that runs in $T_{\Decode}(n)$ and works with probability $p$. Then, for any integer $k$, $\mathcal{C}$ can be $(\alpha, l, k\cdot L(n))$-list-recovered in $k T_{\Decode}(n)$ time with $1-(1-p)^k$ success probability.
\end{lemma}
\begin{proof}
Use a decoding procedure for $\mathcal{C}$ that repeats the given decoder $k$ times and outputs the union of the lists produced by them. The final list size will be at most $kL(n)$ long, the running time will be $k T_{\Decode}(n)$, and the failure probability, i.e., the probability of the output list not containing the correct codeword is at most $(1-p)^k$.
\end{proof}

Another required ingredient to the proof of \cref{thm:randomized-labeling} is to show how a probabilistic decoder affect the approximation guarantee of \cref{alg:EnhancedEditDistanceApprox}. To this end, we provide the following lemma as an analogy of \cref{thm:enhanced-edit-distance-approximation-ratio} when the decoder of code $\mathcal{C}_1$ is not deterministic.

\begin{lemma}\label{lem:randomized-decoder-approximation-guarantee}
Let the decoder of code $\mathcal{C}_1$ ($\Decode_{\mathcal{C}_1}(\cdot)$) be a randomized algorithm that $L_1$-list decodes the code $\mathcal{C}_1$ with probability $1-p$. Then, with probability $1-e^{-\frac{2|S'|p}{3N_1}}$, \cref{alg:EnhancedEditDistanceApprox} will generate a set of up to $(1+23\eps)\ED(S, S')+2p|S'|$ insertions and deletions that turn $S$ into $S'$.
\end{lemma}
\begin{proof}
If $\Decode_{\mathcal{C}_1}(\cdot)$ worked with probability 1, the outcome of $\mathcal{A}$ would contain up to $(1+23\eps_1)$ insertions and deletions. Each time that $\Decode_{\mathcal{C}_1}$ fails to correctly list-decode a block of length $N_1$ ($C'_i$), up to $N_1$ edges from $\mathcal{M}_{ALG}$ might be lost and, consequently, there can be up to $2N_1$  units of increase in the number of insertions and deletions generated by $\mathcal{A}$.

There are a total of $n=|S'|/N_1$ list decodings and each might fail with probability $p$. Using the Chernoff bound, 
$$\Pr(\text{more than $2np$ failures})\le e^{-2np/3} = e^{-\frac{2|S'|p}{3N_1}}.$$
Thus, with probability $1-e^{-\frac{2|S'|p}{3N_1}}$, the output of $\mathcal{A}$ contains 
$(1+23\eps)\ED(S, S')+2npN_1=(1+23\eps)\ED(S, S')+2p|S'|$
or less insertions and deletions.
\end{proof}

We are now adequately equipped to prove \cref{thm:randomized-labeling}.

\subsection{Proof of \cref{thm:randomized-labeling}}
\begin{proof}
Our construction closely follows the steps taken in the proof of \cref{thm:enhanced-labeling}. Let $\eps'=\eps_1/46$. Take $\mathcal{C}_1$ from the \cref{thm:HemenwayEtAl} with parameters $\eps_{\mathcal{C}_1}=\eps'$, $\rho_{\mathcal{C}_1}=2\eps'$, $l_{\mathcal{C}_1}=6/\eps'$, $c_{\mathcal{C}_1}=\eps_0$, and block length $N_1$ where $N_1$ is large enough so that $N_1\cdot q_{\mathcal{C}_1}^{\rho_{\mathcal{C}_1}/2\cdot N_1} \geq n$ where $q_{\mathcal{C}_1}$ is the size of the alphabet of the family codes.

According to~\cref{thm:HemenwayEtAl}, $\mathcal{C}_1$ is probabilistically $(\eps', \eps'/6, L(N_1))$-list recoverable in $O_{\eps_1}(N_1^{1+\eps_0})$ time where $L(N_1)=\exp(\exp(\exp(\log^* N_1)))$ and success probability is 2/3. We use \cref{lem:probability-enhancement} with repetition number parameter $k=\log_3 \frac{2}{\eps_2}$ to obtain a 
\noSTOC{$\left(\eps', \frac{\eps'}{6}, O\left(\log \frac{1}{\eps_2}L(N_1)\right)\right)$}\STOConly{$\left(\eps', \eps'/6, O\left(\log (1/\eps_2)\cdot L(N_1)\right)\right)$}-list recovery algorithm for $\mathcal{C}_1$ that succeeds with probability $1-(\frac{1}{3})^k = 1-\frac{\eps_2}{2}$ and runs in 
\noSTOC{$O_{\eps_1}\left(\frac{N_1^{1+\eps_0}}{\eps_2}\right)$} 
\STOConly{$O_{\eps_1}\left(N_1^{1+\eps_0} / \eps_2\right)$} 
time.

We now use \cref{thm:coding_applications} with parameters 
$\eps_{\text{conv}}=\eps'$ and $\gamma_{\text{conv}}=1-2\eps'$
to convert list-recoverable code $\mathcal{C}_1$ to a list-decodable insertion-deletion code $\tilde{\mathcal{C}_1}$ by indexing its codewords with an appropriately chosen indexing sequence from \cref{thm:near-linear-labeling} and a synchronization string. Note that we can do this conversion using \cref{thm:coding_applications} since $\gamma_{\text{conv}} = 1-2\eps' \le \frac{l_{\mathcal{C}_1} \cdot \eps_{\text{conv}}}{3}-1
= \frac{6/\eps' \cdot \eps'}{3}-1=1$. Also, $\tilde{\mathcal{C}_1}$ can $ O\left(\log \frac{1}{\eps_2}L(N_1)\right)$-list decode from any $\gamma_{\text{conv}}=1-2\eps'$ fraction of insertions and any $1-\alpha_{\mathcal{C}_1}-\eps_{\text{conv}} = 1-\frac{\eps}{46}-\eps' = 1-2\eps'$ fraction of deletions in $O_{\eps_1, \eps_2}\left(N_1^{1+\eps_0} + N_1\poly(\log N_1)\right)$.

We further take code $\tilde{\mathcal{C}_2}$ from \cref{cor:capacity_approaching_list_decodable_insdel} with parameter $\eps_{\tilde{\mathcal{C}_2}}=2\eps'$ and block length $N_2$ large enough so that $N_2\cdot q_{\tilde{\mathcal{C}_2}}^{\eps'/2\cdot N_2} \geq N_1$. $\tilde{\mathcal{C}_2}$ is $\exp(\exp(\exp(N_2)))$-list decodable from any $1-2\eps'$ fraction of insertions and $1-2\eps'$ fraction of deletions.

String $I$ for the indexing scheme is constructed according to the structure described in \cref{sec:two-layers} using $\tilde{\mathcal{C}_1}$ and $\tilde{\mathcal{C}_2}$.

We define algorithm $\mathcal{\widetilde{\ED}}_I$ as an algorithm that takes $S\times I$ and $S'$ and runs an instance of \cref{alg:EnhancedEditDistanceApprox} with $S\times I$, $S'$, $N_1$, $N_2$, and decoders of $\tilde{\mathcal{C}_i}$ as its input. 
As codes $\tilde{\mathcal{C}_i}$ list decode from $1-2\eps'$ fraction of insertions and deletions, \cref{lem:randomized-decoder-approximation-guarantee} guarantees that $\mathcal{\widetilde{\ED}}_I$ generates a set of insertions and deletions over $S\times I$ that converts it to $S'$ and is of size $(1+23\cdot2\eps')\ED(S\times I, S') + 2\cdot \frac{\eps_2}{2} |S'|
=(1+\eps)\ED(S\times I, S') + \eps_2 |S'|$ or less with probability 
$1-e^{-\frac{\eps_2}{3N_1}}=1-e^{-O\left(\frac{\eps_2}{\log n}\right)}=1-\frac{1}{n^{O_{\eps_1, \eps_2}(1)}}$. 

Finally, since $N_1 = O(\log n)$, $N_2=O(\log\log n)$, $\tilde{\mathcal{C}_1}$ is list-decodable in $O(N_1^{1+\eps_0} + N_1\poly(\log N_1))$ time and $\tilde{\mathcal{C}_2}$ is efficiently list-decodable, \cref{thm:enhanced-edit-distance-approximation-time} guarantees that $\mathcal{\widetilde{\ED}}_I$ runs in 
\noSTOC{\begin{eqnarray*}
&&O_{\eps_1,\eps_2}\left(\frac{n}{N_1}\cdot T_{\Decode_{\tilde{\mathcal{C}_1}}}(N_1) + \frac{n}{N_2}\cdot T_{\Decode_{\tilde{\mathcal{C}_2}}}(N_2) + N_2L_1L_2\cdot n\log\log n\right)\\ 
&=&O_{\eps_1,\eps_2}\left(\frac{n}{\log n}\cdot T_{\Decode_{\tilde{\mathcal{C}_1}}}(\log n) + \frac{n}{\log \log n}\cdot T_{\Decode_{\tilde{\mathcal{C}_2}}}(\log \log n) + n\log^2 \log n L_{\tilde{\mathcal{C}_1}}(N_1)L_{\tilde{\mathcal{C}_2}}(N_2)\right)\\ 
&=& O_{\eps_1,\eps_2}\bigg(\frac{n}{\log n}\cdot
\left[\log^{1+\eps_0} n+\log n\poly(\log\log n)\right]+\frac{n}{\log\log n}\cdot\left[\poly(\log \log n)\right]+\\
&&\quad\quad\quad n\log^2\log n \cdot \exp(\exp(\exp(\log^* n)))
\bigg)\\
&=&O_{\eps_1,\eps_2}\left(n\log^{\eps_0} n\right)
\end{eqnarray*}}
\STOConly{\begin{eqnarray*}
&&O_{\eps_1,\eps_2}\left(\frac{nT_{\Decode_{\tilde{\mathcal{C}_1}}}(N_1)}{N_1} + \frac{nT_{\Decode_{\tilde{\mathcal{C}_2}}}(N_2)}{N_2}  + N_2L_1L_2\cdot n\log\log n\right)\\ 
&=&O_{\eps_1,\eps_2}\bigg(\frac{n}{\log n}\cdot T_{\Decode_{\tilde{\mathcal{C}_1}}}(\log n) + \frac{n}{\log \log n}\cdot T_{\Decode_{\tilde{\mathcal{C}_2}}}(\log \log n)\\
&& + n\log^2 \log n L_{\tilde{\mathcal{C}_1}}(N_1)L_{\tilde{\mathcal{C}_2}}(N_2)\bigg)\\ 
&=&O_{\eps_1,\eps_2}\left(n\log^{\eps_0} n\right)
\end{eqnarray*}}
 time.
\end{proof}


\section{Near-Linear Time \noSTOC{Insertion-Deletion}\STOConly{InsDel} Codes}\label{sec:applications}
The construction of efficient (uniquely-decodable) insertion-deletion codes from \cite{haeupler2017synchronization} and list-decodable codes from \cite{haeupler2018synchronization4} profoundly depend on decoding synchronization strings that are attached to codewords of an appropriately chosen Hamming-type code. The decoding procedure, which was introduced in \cite{haeupler2017synchronization}, consists of multiple rounds of computing the longest common subsequence (LCS) between a synchronization string and a given string. In this section, we will show that using the indexing schemes that are introduced in this paper, one can compute approximations of the LCSs instead of exact LCSs to construct insertion-deletion codes of similar guarantees as in \cite{haeupler2017synchronization, haeupler2018synchronization4} that have faster decoding complexity.

Specifically, for uniquely-decodable insertion-deletion codes, \cite{haeupler2017synchronization} provided codes with linear encoding-time and quadratic decoding-time that can approach the singleton bound, i.e., for any $0<\delta<1$ and $0<\eps<1-\delta$ can correct from $\delta$-fraction of insertions and deletions and achieve a rate of $1-\delta-\eps$. Further, same authors~\cite{haeupler2017synchronization3} provided codes with linear encoding complexity and near-linear decoding complexity can can correct from $\delta<1/3$ fraction of insertions and deletions but only achieve a rate of $1-3\delta-\eps$. In \cref{thm:coding_applications_unique} we will provide insertion-deletion codes that give the best of the two worlds, i.e., approach the Singleton bound and can be decoded in near-linear time.

Further, in \cref{thm:coding_applications}, we show that the same improvement can be made over list-decodable insertion-deletion codes of \cite{haeupler2018synchronization4}. However, this improvement brings downs the complexity of all components of the decoding procedure to near-linear time except the part that depends on the decoding of a list-recoverable code that is used as a black-box in the construction from \cite{haeupler2018synchronization4}. Even though this progress does not immediately improve the decoding time of list-decodable codes of \cite{haeupler2018synchronization4}, it opens the door to enhancement of the decoding complexity down to potentially a near-linear time by the future advances in the design of list-recoverable codes.


\subsection{Enhanced Decoding of Synchronization Strings via Indexing}
Synchronization strings, as introduced in \cite{haeupler2017synchronization}, are strings that satisfy \cref{def:synchronization-strings}. Let $S$ be an $\eps$-synchronization string that is communicated through a channel that suffers from a certain fraction of insertions and deletions. A decoding algorithm $\Decode_S$ for synchronization string $S$ under such channel is an algorithm that takes string that is arrived at the receiving end of the channel, and for each symbol of that string, guesses its actual position in $S$. We measure the quality of the decoding algorithm $\Decode_S$ by a metric named as \emph{misdecodings}. A misdecoding in the above-mentioned decoding procedure is a symbol of $S$ that (1) is not deleted by the channel \emph{and} (2) is not decoded correctly by $\Decode_S$. (formal definitions in \cite{haeupler2017synchronization})

The important quality of synchronization strings that is used in the design of insertion-deletion codes in \cite{haeupler2017synchronization, haeupler2018synchronization4} is that there are decoders for any $\eps$-synchronization string that run in quadratic time $O(n^2/\eps)$ and guarantee $O(n\sqrt{\eps})$ misdecodings. In this paper, by indexing synchronization strings with indexing sequences introduced in \cref{thm:near-linear-labeling}, we will show that one can obtain a near-linear decoding that provides similar misdecoding guarantee.

In the rest of this section, we first present and prove a theorem that shows an indexed synchronization string can be decoded in near-linear time with guarantees that are expected in Theorem 6.14 of \cite{haeupler2017synchronization} and Lemma 3.2 of \cite{haeupler2018synchronization4}. We then verify that the steps taken in \cite{haeupler2017synchronization, haeupler2018synchronization4} still follow through.

\begin{theorem}\label{lem:enhanced-sync-string-decoding}
Let $S$ be a string of length $n$ that consists of the coordinate-wise concatenation of an $\eps_{s}$-synchronization string and an $\eps_I$-indexing sequence from \cref{thm:near-linear-labeling}. Assume that $S$ goes through a channel that might impose up to $\delta \cdot n$ deletions and $\gamma\cdot n$ symbol insertions on $S$ for some $0\le\delta<1$ and $0\le\gamma$ and arrives as $S'$ on the receiving end of the channel. For any positive integer $K$, there exists a decoding for $S'$ that runs in $O(Kn\poly(\log n))$ time, guarantees up to 
$n\left(
\frac{1+\gamma}{K(1+\eps_I)}+ \frac{\eps_I(1+\gamma/2)}{1+\eps_I} +K\eps_s
\right)$
 misdecodings, and does not decode more than $K$ received symbol to any number in $[1, n]$.
\end{theorem}
Before proceeding to the proof of \cref{lem:enhanced-sync-string-decoding}, we present and prove the following simple yet useful lemma.
\begin{lemma}\label{lem:LCS-approximation}
Let us have a set of insertions and deletions that converts string $S_1$ to string $S_2$ which is of size $\ED_{APP} \le (1+\eps)\ED(S_1, S_2)$. The common subsequence between $S_1$ and $S_2$ that is implied by such a set ($\LCS_{APP}$) is of size $(1+\eps)|\LCS| - \frac{\eps}{2} (|S_1| + |S_2|)$ or larger.
\end{lemma}
\begin{proof}
\begin{eqnarray*}
|\LCS_{APP}| &=& \frac{|S_1| + |S_2| - \ED_{APP}}{2}\\
&\ge& \frac{|S_1| + |S_2| - (1+\eps)\ED(S_1, S_2)}{2}\\
\noSTOC{&=& \frac{|S_1| + |S_2| - (1+\eps)(|S_1| + |S_2| - 2|\LCS|)}{2}\\}
&=&(1+\eps)|\LCS| - \frac{\eps}{2} \left(|S_1| + |S_2|\right)
\end{eqnarray*}
\end{proof}

\begin{proof}[Proof of \cref{lem:enhanced-sync-string-decoding}]
The global decoding algorithm introduced in \cite{haeupler2017synchronization} and used in \cite{haeupler2018synchronization4}, consists of $K$ repetitions of the following steps:
\begin{enumerate}
\item Find the longest common subsequence (LCS) of $S$ and $S'$.
\item \label{item:decoding-assignment} For any pair $(S[i], S'[j])$ in the LCS, decode $S'[j]$ as $i$th sent symbol.
\item Remove all members of the LCS from $S'$ (not in $S$).
\end{enumerate}
Finally, the algorithm declares a special symbol $\bot$ as the decoded position of all elements of $S'$ that are not included in any of the $K$ LCSs.

To derive a decoding algorithm as promised in the statement of this lemma, we implement similar steps except we make use of the indexing scheme and compute an approximation of LCS instead of the LCS itself. This crucial step reduces the quadratic time required in the global decoding from \cite{haeupler2017synchronization} to near-linear time.

In \cite{haeupler2017synchronization}, it has been shown that any assignment from \cref{item:decoding-assignment} that is derived from any common subsequence between $S$ and $S'$ (not necessarily a LCS) does not contain more than $n\eps_s$ misdecodings, i.e., successfully transmitted symbols of $S$ that are decode incorrectly. (see \ref{lem:self-matching}). Therefore, after $K$ repetitions, among symbols of $S$ that are not deleted, there are at most $Kn\eps_s$ ones that are decoded incorrectly.

To find an upper bound for the misdecodings of this algorithm, we need to bound above the number of successfully transmitted symbols that are not included in any LCS, i.e., decoded as $\bot$ as well. Let $r$ be number of successfully transmitted symbols of $S$ that remain undecoded after $K$ repetitions of the matching procedure described above. Note that these symbols form a LCS of length $r$ between $S$ and the remainder of $S'$ after all symbol eliminations throughout $K$ repetitions.  Indeed, this implies that the size of the LCS at the beginning of each repetition is at least $r$. Therefore, by \cref{lem:LCS-approximation}, the size of the approximate longest common sequence found in each matching is at least 
$(1+\eps_I)r-\eps_I/2(|S|+|S'|) \ge (1+\eps_I)r-\eps_I n(1+\gamma/2)$. Note that sum of the size of all $K$ common subsequences plus the remaining vertices cannot exceed $|S'| \le (1+\gamma)n$. Therefore,
\begin{eqnarray}
&&K\cdot \left[(1+\eps_I)r-\eps_I n(1+\gamma/2)\right] \le (1+\gamma)n\nonumber\\
&\Rightarrow& r \le n\cdot\left[\frac{1+\gamma}{K(1+\eps_I)}+\frac{\eps_I(1+\gamma/2)}{1+\eps_I}\right]\label{eqn:unmatched}
\end{eqnarray}
Using \eqref{eqn:unmatched} along with the fact that there are at most $Kn\eps_s$ incorrectly decoded symbols of $S$ gives that the overall number of misdecodings is at most
$n\cdot\left[\frac{1+\gamma}{K(1+\eps_I)}+\frac{\eps_I(1+\gamma/2)}{1+\eps_I} + K\eps_s\right]$.

Further, as algorithm consists of $K$ computations of the approximated longest common subsequence as described in \cref{sec:labeling}, the running time complexity is $O(Kn\poly(\log n))$.

Finally, note that in each of the $K$ rounds, there is at most one element that gets decoded as each number in $[1, n]$. Therefore, throughout the course of the algorithm, for each $i\in[1, n]$, there are at most $K$ elements of $S'$ that are decoded as $i$.
\end{proof}

\subsection{Near-Linear Time (Uniquely-Decodable) Insertion-Deletion Codes}\label{sec:unique-codes}
The construction of Singleton-bound-approaching uniquely-decodable insertion-deletion codes of \cite{haeupler2017synchronization} is consisted of a Singleton approaching error correcting code and a synchronization string. More precisely, for a given $\delta$ and $\eps$ and a sufficiently large $n$, \cite{haeupler2017synchronization} takes a synchronization string $S$ of length $n$ and a Singleton-bound-approaching error correcting code $\mathcal{C}$ with block length $n$ (from \cite{guruswami2005linear}) and indexes each codeword of $\mathcal{C}$, symbol by symbol, with symbols of $S$. If $S$ is over alphabet $\Sigma_S$ and $\mathcal{C}$ is over alphabet $\Sigma_\mathcal{C}$, the resulting code would be over $\Sigma_\mathcal{C}\times\Sigma_S$.

As for the decoding procedure, note that the input of the decoder is some code word of $\mathcal{C}$, indexed with $S$, that might be altered by up to $\delta\cdot n$ insertions and deletions. Such insertions and deletions might remove some symbols, adds some new ones, or shift the position of some of them. The decoder uses the synchronization portion of each symbol to guess its actual position (in the codeword prior to $n\cdot\delta$ insertions and deletions) and then uses the decoder of code $\mathcal{C}$ to figure out the sent codeword. 

Before proceeding to the proof of \cref{thm:coding_applications_unique}, we represent the following useful theorem from \cite{haeupler2017synchronization}.

\begin{theorem}[Implied by Theorem 4.2 from \cite{haeupler2017synchronization}]\label{thm:modular-code-construction-unique}
Given a synchronization string $S$ over alphabet $\Sigma_S$, an (efficient) decoding algorithm $\mathcal{D}_S$ with at most $k$ misdecodings and decoding complexity $T_{\mathcal{D}_{S}}(n)$ and an (efficient) ECC $\mathcal{C}$ over alphabet $\Sigma_{\mathcal{C}}$ with rate $R_{\mathcal{C}}$, encoding complexity $T_{\mathcal{E}_{\mathcal{C}}}$, and decoding complexity $T_{\mathcal{D}_{\mathcal{C}}}$ that corrects up to $n\delta + 2k$ half-errors, one obtains an insdel code that can be (efficiently) decoded from up to $n\delta$ insertions and deletions. The rate of this code is at least
\noSTOC{$$\frac{R_{\mathcal{C}}}{1 + \frac{\log |\Sigma_S|}{\log |\Sigma_{\mathcal{C}}|}}
.$$}\STOConly{$\frac{R_{\mathcal{C}}}{1 + \log |\Sigma_S| / \log |\Sigma_{\mathcal{C}}|}
$.}
The encoding complexity remains $T_{\mathcal{E}_{\mathcal{C}}}$, the decoding complexity is $T_{\mathcal{D}_{\mathcal{C}}} + T_{\mathcal{D}_{S}}(n)$ and the complexity of constructing the code is the complexity of constructing $\mathcal{C}$ and $S$.
\end{theorem}
\STOConly{\vspace{-2mm}}
We make use of \cref{thm:modular-code-construction-unique} from \cite{haeupler2017synchronization} along with \cref{lem:enhanced-sync-string-decoding} to prove \cref{thm:coding_applications_unique}.

\STOConly{\vspace{-2mm}}
\begin{proof}[\textbf{Proof of \cref{thm:coding_applications_unique}}]
As described earlier in this section, we construct this code by taking an error correcting code that approaches the Singleton bound and then index its codewords with symbols of an $\eps_s$-synchronization string and an indexing scheme from \cref{thm:near-linear-labeling} with parameter $\eps_I$.
For a given $\delta$ and $\eps$, we choose 
$\eps_I=\frac{\eps}{18}$, 
$\eps_s=\frac{\eps^2}{288}$.
Furthermore, we use the decoding algorithm from \cref{lem:enhanced-sync-string-decoding} with repetition parameter $K=\frac{24}{\eps}$. With $\eps_s, \eps_I$, and $K$ chosen as such, the decoding algorithm guarantees a misdecoding count of 
$n\cdot\left[\frac{1+\gamma}{K(1+\eps_I)}+\frac{\eps_I(1+\gamma/2)}{1+\eps_I} + K\eps_s\right] \le 
n\cdot\left[\frac{\eps}{12} + \frac{\eps}{12} + \frac{\eps}{12}\right]=\frac{n\eps}{4}$
or less. (note that there can be up to $\delta n$ insertions, i.e., $\gamma \le \delta < 1$)

It has been shown in \cite{haeupler2017synchronization3} that such synchronization string can be constructed in linear time over an alphabet of size $\eps_s^{-O(1)}$. Also, the indexing sequence from \cref{thm:near-linear-labeling} has an alphabet of size \noSTOC{$\exp\left(\frac{\log (1/\eps_I)}{\eps_I^3}\right)$}\STOConly{$\exp\left(\log (1/\eps_I) / \eps_I^3\right)$}. Therefore, the alphabet size of the coordinate-wise concatenation of the $\eps_s$-synchronization string and the indexing sequence is \noSTOC{$|\Sigma_S|=\exp\left(\frac{\log (1/\eps)}{\eps^3}\right)$}\STOConly{$|\Sigma_S|=\exp\left(\log (1/\eps) / \eps^3\right)$}.

As the next step, we take code $\mathcal{C}$ from \cite{guruswami2005linear} as a code with distance $\delta_{\mathcal{C}} = \delta + \frac{\eps}{2}$ and rate $1-\delta_{\mathcal{C}}-\frac{\eps}{4}$ over an alphabet of size $|\Sigma_{\mathcal{C}}|=|\Sigma_S|^{4/\eps}$. Note that 
\noSTOC{$|\Sigma_S|=\exp\left(\frac{\log (1/\eps)}{\eps^3}\right)$}\STOConly{$|\Sigma_S|=\exp\left(\log (1/\eps) / \eps^3\right)$}, therefore, the choice of $|\Sigma_{\mathcal{C}}|$ is large enough to satisfy the requirements of \cite{guruswami2005linear}.
$\mathcal{C}$ is also encodable and decodable in linear time.

Plugging $\mathcal{C}$ and $S$ as described above in \cref{thm:modular-code-construction-unique} gives an insertion-deletion code that can be encoded in linear time, be decoded in $O(Kn\poly(\log n))$ time, corrects from any $\delta n$ insertions and deletions, achieves a rate of
\noSTOC{$\frac{R_{\mathcal{C}}}{1 + \frac{\log |\Sigma_S|}{\log |\Sigma_{\mathcal{C}}|}}\ge \frac{1-\delta-3\eps/4}{1+\eps/4}\ge 1-\delta-\eps$}\STOConly{$\frac{R_{\mathcal{C}}}{1 + \log |\Sigma_S| / \log |\Sigma_{\mathcal{C}}|}\ge \frac{1-\delta-3\eps/4}{1+\eps/4}\ge 1-\delta-\eps$}, and is over an alphabet of size \STOConly{$\exp\left(\log (1/\eps) / \eps^4\right)$}\noSTOC{$\exp\left(\frac{\log (1/\eps)}{\eps^4}\right)$}.
\end{proof}

\subsection{Improved List-Decodable \noSTOC{Insertion-Deletion}\STOConly{InsDel} Codes}\label{sec:list-decodable-codes}
A very similar improvement is also applicable to the design of list-decodable insertion-deletion codes from~\cite{haeupler2018synchronization4} as it also utilizes indexed synchronization strings and a similar position recovery procedure.  In the following theorem, we will provide a black-box conversion of a given list-recoverable code to a list-decodable insertion-deletion code that only adds a near-linear time overhead to the decoding complexity. Hence, the following theorem paves the way to obtaining insertion-deletion codes that are list-decodable in near-linear time upon the design of near-linear time list-recoverable codes. We will use the following theorem to prove \cref{thm:list-decodable-codes} at the end of this section.

\begin{theorem}\label{thm:coding_applications}
Let $\mathcal{C}:\Sigma^{nR}\rightarrow\Sigma^n$ be a $(\alpha, l, L)$-list recoverable code with rate $R$, encoding complexity $T_{Enc}(\cdot)$ and decoding complexity complexity $T_{Dec}(\cdot)$. For any $\eps>0$ and $\gamma \leq \frac{l\eps}{3}-1$, by indexing codewords of $\mathcal{C}$ with an 
$\eps_s=\frac{\eps^2}{9(1+\gamma)}$-synchronization string over alphabet $\Sigma_{s}$ and 
$\eps_I=\frac{\eps}{3(1+\gamma/2)}$-indexing sequence over alphabet $\Sigma_I$, one can obtain an $L$-list decodable insertion-deletion code $\mathcal{C}': \Sigma^{nR}\rightarrow[\Sigma\times\Sigma_s\times\Sigma_I]^n$ that corrects from $\delta < 1-\alpha-\eps$ fraction of deletions and $\gamma$ fraction of insertions. $\mathcal{C}'$ is encodable and decodable in $O(T_{Enc}(n)+n)$ and 
$O_{\eps, \gamma}	(T_{Dec}(n)+n\poly(\log n))$
 time respectively.
\end{theorem}

\begin{proof}
We closely follow the proof of Theorem 3.1 from \cite{haeupler2018synchronization4} except that we use an indexed synchronization string to speed up the decoding procedure.

Index the code $\mathcal{C}$ with an $\eps_s=\frac{\eps^2}{9(1+\gamma)}$-synchronization string and an $\eps_I=\frac{\eps}{3(1+\gamma/2)}$-indexing sequence as constructed in \cref{thm:near-linear-labeling} to obtain code $\mathcal{C}'$.

In the decoding procedure, for a given word $\tilde x$ that is $\delta n$ deletions and $\gamma n$ insertions far from some codeword $x\in\mathcal{C}'$, we first use the decoding algorithm from \cref{lem:enhanced-sync-string-decoding} to decode the index portion of symbols with parameter \STOConly{$K=3(1+\gamma) / \eps$}\noSTOC{$K=\frac{3(1+\gamma)}{\eps}$}.
This will give a list of up to \noSTOC{$K=\frac{3(1+\gamma)}{\eps} \leq l$}\STOConly{$K=3(1+\gamma)/\eps \leq l$} candidate symbols for each position of the codeword $x$.

We know from \cref{lem:enhanced-sync-string-decoding} that all but
\noSTOC{$$
n\left(
\frac{1+\gamma}{K(1+\eps_I)}+ \frac{\eps_I(1+\gamma/2)}{1+\eps_I} +K\eps_s
\right)\\
\le n\left(
\frac{\eps}{3(1+\eps_I)}+ \frac{\eps}{3(1+\eps_I)} +\frac{\eps}{3}
\right)\le  n\eps
$$}\STOConly{$$
n\left(
\frac{1+\gamma}{K(1+\eps_I)}+ \frac{\eps_I(1+\gamma/2)}{1+\eps_I} +K\eps_s
\right)\\
\le n\left(
\frac{\eps}{3(1+\eps_I)}\times 2 +\frac{\eps}{3}
\right)\le  n\eps
$$}
of the symbols of $x$ that are not deleted are in the correct list. As there are up to $n(1-\delta)$ deleted symbols, all but $n(1-\delta-\eps) > n\alpha$ of the lists contain the symbol from the corresponding position in $x$. Having such lists, the receiver can use the list-recovery function of $\mathcal{C}$ to obtain an $L$-list-decoding for $\mathcal{C}'$. 

The encoding complexity follows from the fact that synchronization strings  be constructed in linear time~\cite{haeupler2017synchronization3, cheng2018synchronization}, the decoding complexity follows from 
\cref{lem:enhanced-sync-string-decoding}, and the alphabet of $\mathcal{C}'$ is trivially $\Sigma\times\Sigma_s\times\Sigma_I$ as it is obtained by indexing codewords of $\mathcal{C}$ with the $\eps_s$-synchronization string and the $\eps_I$-indexing sequence.
\end{proof} 
\STOConly{\vspace{-2mm}}
We now use \cref{thm:coding_applications} to prove \cref{thm:list-decodable-codes}.
\STOConly{\vspace{-1mm}}
\begin{proof}[\textbf{Proof of \cref{thm:list-decodable-codes}}]
Take list-recoverable code $\mathcal{C}$ from~\cref{thm:HemenwayEtAl} with parameters $\rho_{\mathcal{C}}=1-\delta-\frac{\eps}{2}$, $\eps_{\mathcal{C}}=\frac{\eps}{4}$, $l_{\mathcal{C}}=\frac{12\gamma+4}{\eps}$, and $c_{\mathcal{C}}=\eps_0$ over an alphabet $\Sigma$ of adequately large size $|\Sigma|\geq q_{0, \mathcal{C}}$ which we determine later. According to~\cref{thm:HemenwayEtAl}, $\mathcal{C}$ has a rate of $\rho_\mathcal{C}$ and a randomized $\left(\rho_{\mathcal{C}}+\eps_{\mathcal{C}}, l_{\mathcal{C}}, L(n)=\exp(\exp(\exp(\log^* n)))\right)$-list recovery that works in $O(n^{1+\eps_0})$ time and succeeds with probability 2/3.

We plug code $\mathcal{C}$ into \cref{thm:coding_applications} with parameters $\eps_{\text{conv}}=\frac{\eps}{4}$ and $\gamma_{\text{conv}}=\gamma$ to obtain code $\mathcal{C}'$. We can do this because $\gamma_{\text{conv}}\le\frac{l_{\mathcal{C}}\eps_{conv}}{3}-1$. According to \cref{thm:coding_applications}, $\mathcal{C}'$ is $L(n)$-list decodable from $1-\rho_{\mathcal{C}}-\eps_{\mathcal{C}}-\eps_{\text{conv}}=\delta$ fraction of deletions and $\gamma$ fraction of insertions in $O(n^{1+\eps_0})$. This list-decoding is correctly done if the list-recovery algorithm works correctly. Therefore, the list decoder succeeds with probability 2/3 or more.

Note that the $\eps_s$-synchronization strings in \cref{thm:coding_applications} exist over alphabets of size $|\Sigma_s|=\eps_s^{-O(1)}$ and $\eps_I$-indexing sequence exist over alphabets of size 
\noSTOC{$|\Sigma_I|=\exp\left(\frac{\log (1/\eps_I)}{\eps_I^3}\right)$.}\STOConly{$|\Sigma_I|=\exp\left(\log (1/\eps_I) / \eps_I^3\right)$.}
Therefore, if we take alphabet $\Sigma$ large enough so that 
\noSTOC{$|\Sigma|\ge\max\left\{|\Sigma_s\times\Sigma_I|^{2/\eps}, q_{0, \mathcal{C}}\right\}= \max\left\{\exp\left(\frac{\log (1/\eps)}{\eps^4}\right),q_{0, \mathcal{C}}\right\} = O_{\eps_0, \eps, \gamma}(1)$}
\STOConly{$|\Sigma|\ge\max\left\{|\Sigma_s\times\Sigma_I|^{2/\eps}, q_{0, \mathcal{C}}\right\}= \max\left\{\exp\left(\log (1/\eps) / \eps^4\right),q_{0, \mathcal{C}}\right\} = O_{\eps_0, \eps, \gamma}(1)$}
the rate of the resulting code will be
\noSTOC{$$\frac{\rho_{\mathcal{C}}}{1 + \frac{\log \left(|\Sigma_S|\times|\Sigma_l|\right)}{\log |\Sigma|}}\ge \frac{1-\delta-\eps/2}{1+\eps/2}\ge1-\delta-\eps.$$}\STOConly{$\frac{\rho_{\mathcal{C}}}{1 + \log \left(|\Sigma_S|\times|\Sigma_l|\right)/\log |\Sigma|}\ge \frac{1-\delta-\eps/2}{1+\eps/2}\ge1-\delta-\eps.$}

Finally, the encoding and decoding complexities directly follow from \cref{thm:coding_applications}.
\end{proof}

\bibliographystyle{plain}
\bibliography{bibliography}

\end{document}